\documentclass[atmp]{ipart_v1}

\Vol{xx}
\Issue{x}
\Year{2020}
\firstpage{000}

\usepackage{t1enc}
\usepackage[latin1]{inputenc}
\usepackage[english]{babel}

\usepackage{amsthm}
\usepackage{yfonts}

\usepackage{bbm}
\usepackage{bm}
\usepackage{mathrsfs}
\usepackage{graphicx}

\newcommand{\be}[0]{\begin{equation}}
\newcommand{\ee}[0]{\end{equation}}

\numberwithin{equation}{section}

\theoremstyle{plain}
\newtheorem{theorem}{Theorem}[section]
\newtheorem{lemma}[theorem]{Lemma}

\begin{document}

\title[On hyperspherical associated Legendre functions]{On hyperspherical associated Legendre functions: the extension of spherical harmonics to $N$ dimensions}

\author[L. M. B. C. Campos and M. J. S. Silva]{L. M. B. C. Campos and M. J. S. Silva}

\begin{abstract}
The solution in hyperspherical coordinates for $N$ dimensions is given for a general class of partial differential equations of mathematical physics including the Laplace, wave, heat and Helmholtz, Schr\"{o}dinger, Klein-Gordon and telegraph equations and their combinations. The starting point is the Laplacian operator specified by the scale factors of hyperspherical coordinates. The general equation of mathematical physics is solved by separation of variables leading to the dependencies: (i) on time by the usual exponential function; (ii) on longitude by the usual sinusoidal function; (iii) on radius by Bessel functions of order generally distinct from cylindrical or spherical Bessel functions; (iv) on one latitude by associated Legendre functions; (v) on the remaining latitudes by an extension, namely the hyperspherical associated Legendre functions. The original associated Legendre functions are a particular case of the Gaussian hypergeometric functions, and the hyperspherical associated Legendre functions are also a more general particular case of the Gaussian hypergeometric functions so that it is not necessary to consider extended Gaussian hypergeometric functions.
\end{abstract}


\maketitle

\section{Introduction}
\label{sec:1}

One of the classical problems of analysis, with a multitude of physical applications in fluids \cite{Lamb1945}, solids \cite{Love1944}, electromagnetism \cite{Stratton1941}, acoustics \cite{Rayleigh1945} and quantum mechanics \cite{Landau1965}, is the solution of the Laplace and related equations in terms of spherical harmonics \cite{Forsyth1956,MacRobert1967,Whittaker1996,Hobson1931,Erdelyi1953}. The present paper considers the extension to $N$ dimensions using hyperspherical coordinates, consisting of one radial distance, one longitude and $N-2$ latitudes (section \ref{sec:2}); the transformation from $N$-dimensional Cartesian coordinates (subsection \ref{sec:2.1}) proves that hyperspherical coordinates form an orthogonal curvilinear system, and the scale factors (subsection \ref{sec:2.2}) specify the Laplace operator (subsection \ref{sec:2.3}). The latter is used in a generic equation of mathematical physics (section \ref{sec:3}) combining the Laplace, wave, heat, Helmholtz, Schr\"{o}dinger, Klein-Gordon and telegraph equations (subsection \ref{sec:3.1}). The solution by separation of variables (subsection \ref{sec:3.2}) leads to the usual exponential dependence on time and sinusoidal dependence on longitude, and: (i) the radial dependence involves Bessel functions, generally of order distinct from cylindrical and spherical Bessel functions (subsection \ref{sec:3.3}); (ii) only one latitudinal dependence is specified by associated Legendre functions, and all others involve a generalization, namely the hyperspherical associated Legendre functions (section \ref{sec:4}). The latter (subsection \ref{sec:4.1}) first appears when using hyperspherical harmonics in four dimensions (subsection \ref{sec:4.2}), and one more arise in the general solution of the generic equation of mathematical physics in hyperspherical coordinates for higher dimensions (subsection \ref{sec:4.3}). 

Thus the main new feature (section \ref{sec:5}) is the introduction of hyperspherical associated Legendre functions (subsection \ref{sec:5.1}) which, although more general than the original associated Legendre functions (subsection \ref{sec:5.2}), are also reducible to the Gaussian hypergeometric functions \cite{Forsyth1956,Whittaker1996,Klein1933,Appell1925,Erdelyi1953,Ince1956,Campos2011,Campos2012} (subsection \ref{sec:5.3}); thus they do not require further generalization (section \ref{sec:6}), for example to the extended Gaussian hypergeometric functions \cite{Campos2000a,Campos2001}. The hyperspherical associated Legendre functions introduced in the present paper generalize not only the associated Legendre functions \cite{MacRobert1967,Hobson1931}, but also the hyperspherical Legendre functions \cite{Campos2014}.

\section{Hyperspherical coordinates as a curvilinear orthogonal system in $N$ dimensions}
\label{sec:2}

The relation between hyperspherical and Cartesian coordinates in $N$ dimensions (subsection \ref{sec:2.1}) specifies the base vectors and hence the metric tensors (subsection \ref{sec:2.2}), proving it is an orthogonal curvilinear system. The scale factors can be used to specify the invariant differential operators in hyperspherical coordinates, such as the Laplacian (subsection \ref{sec:2.3}).

\subsection{Transformations between hyperspherical and Cartesian coordinates}
\label{sec:2.1}

The hyperspherical coordinates in $N$ dimensions are defined by the relation with the Cartesian coordinates as a generalization of polar and spherical coordinates. All are orthogonal coordinate systems with straightforward inversion.

The hyperspherical coordinates are an $N$-dimensional generalization of spherical coordinates with one radius, $0\leq r<\infty$, one longitude, $0\leq\phi\leq 2\pi$, and $N-2$ latitudes, $0\leq\theta_1, \ldots, \theta_{N-2}\leq\pi$, defined by the transformation to Cartesian coordinates in $N$ dimensions:
\begin{align}
\begin{aligned}
&x_1=r\cos\theta_1,\\
&x_2=r\sin\theta_1\cos\theta_2,\\
&x_3=r\sin\theta_1\sin\theta_2\cos\theta_3,\\
&\vdots\\
&x_n=r\sin\theta_1\sin\theta_2\ldots\sin\theta_{n-1}\cos\theta_n,\\
&\vdots\\
&x_{N-2}=r\sin\theta_1\sin\theta_2\ldots\sin\theta_{N-3}\cos\theta_{N-2},\\
&x_{N-1}=r\sin\theta_1\sin\theta_2\ldots\sin\theta_{N-2}\cos\phi,\\
&x_{N}=r\sin\theta_1\sin\theta_2\ldots\sin\theta_{N-2}\sin\phi.
\end{aligned}
\label{eq2}
\end{align}

In two dimensions, $N=2$, this corresponds to the transformation from polar coordinates ($r$, $\phi$) to Cartesian coordinates with $x_1\equiv x$ and $x_2\equiv y$; in three dimensions, $N=3$, it corresponds to the transformation from spherical ($r$, $\theta$, $\phi$) to Cartesian coordinates with $x_1\equiv z$ along the polar axis, and with $x_2\equiv x$ and $x_3\equiv y$ transversely. In $N$ dimensions, the radius $r$ and longitude $\phi$ remain, and more latitudes $\theta_1, \ldots, \theta_{N-2}$ are introduced. The name hyperspherical coordinates arises because the first coordinate hypersurface is an hypersphere.

The coordinate transformation inverse to \eqref{eq2}, that is, from $N$-dimensional Cartesian to hyperspherical coordinates, is
\begin{equation}
\begin{aligned}
&r=\left|
\left(x_1\right)^2+\left(x_2\right)^2+ \ldots +\left(x_N\right)^2
\right|^{1/2},\\
&\cot\theta_1=x_1\left|
\left(x_2\right)^2+ \ldots +\left(x_N\right)^2
\right|^{-1/2},\\
&\cot\theta_2=x_2\left|
\left(x_3\right)^2+ \ldots + \left(x_N\right)^2
\right|^{-1/2},\\
&\vdots\\
&\cot\theta_n=x_n\left|
\left(x_{n+1}\right)^2+ \ldots +\left(x_N\right)^2
\right|^{-1/2},\\
&\vdots\\
&\cot\theta_{N-2}=x_{N-2}\left|
\left(x_{N-1}\right)^2+\left(x_N\right)^2
\right|^{-1/2},\\
&\cot\phi=\frac{x_{N-1}}{x_N}
\end{aligned}
\label{eq3}
\end{equation}
where the radius is evaluated through the first equation of \eqref{eq3} and the circular cotangent was used in all others equations of \eqref{eq3}. The coordinate hypersurface $r=\mathrm{const}$ is an hypersphere of radius $r$.

It would be possible to specify the transformation from $N$-dimensional Cartesian to hyperspherical coordinates using in \eqref{eq3} instead of the cotangent, the tangent, sine, cosine, secant or cosecant functions.

\subsection{Base vectors and scale factors}
\label{sec:2.2}

The direct and inverse transformations between $N$-dimensional hyperspherical and Cartesian coordinates specify the scale factors and hence the metric tensor and volume elements.

The Cartesian components of the hyperspherical base vectors follow from the transformation \eqref{eq2}, from hyperspherical to Cartesian coordinates,
\begin{equation}
\begin{aligned}
\overrightarrow{e}_r&=\frac{\partial x_i}{\partial r}=\left\lbrace\cos\theta_1, \sin\theta_1\cos\theta_2, \sin\theta_1\sin\theta_2\cos\theta_3, \ldots,\right.\\
&\sin\theta_1\sin\theta_2\sin\theta_3\ldots\sin\theta_{n-1}\cos\theta_n,\ldots, \sin\theta_1\sin\theta_2\sin\theta_3\ldots\sin\theta_{N-3}\cos\theta_{N-2},\\
&\left.\sin\theta_1\sin\theta_2\sin\theta_3\ldots\sin\theta_{N-2}\cos\phi, \sin\theta_1\sin\theta_2\sin\theta_3\ldots\sin\theta_{N-2}\sin\phi\right\rbrace,\\
\overrightarrow{e}_1&=\frac{\partial x_i}{\partial\theta_1}=r\left\lbrace -\sin\theta_1, \cos\theta_1\cos\theta_2, \cos\theta_1\sin\theta_2\cos\theta_3, \ldots, \right.\\
&\cos\theta_1\sin\theta_2\sin\theta_3\ldots\sin\theta_{n-1}\cos\theta_n, \ldots, \cos\theta_1\sin\theta_2\sin\theta_3\ldots\sin\theta_{N-3}\cos\theta_{N-2},\\
&\left. \vphantom{\sin\theta_{N-2}}\cos\theta_1\sin\theta_2\sin\theta_3\ldots\sin\theta_{N-2}\cos\phi, \cos\theta_1\sin\theta_2\sin\theta_3\ldots\sin\theta_{N-2}\sin\phi\right\rbrace,\\
\overrightarrow{e}_2&=\frac{\partial x_i}{\partial\theta_2}=r\sin\theta_1\left\lbrace 0, -\sin\theta_2, \cos\theta_2\cos\theta_3, \cos\theta_2\sin\theta_3\cos\theta_4, \ldots,\right.\\
&\cos\theta_2\sin\theta_3\sin\theta_4\ldots\sin\theta_{n-1}\cos\theta_n,\ldots, \cos\theta_2\sin\theta_3\sin\theta_4\ldots\sin\theta_{N-3}\cos\theta_{N-2},\\
&\left. \cos\theta_2\sin\theta_3\sin\theta_4\ldots\sin\theta_{N-2}\cos\phi, \cos\theta_2\sin\theta_3\sin\theta_4\ldots\sin\theta_{N-2}\sin\phi\right\rbrace,\\
&\vdots\\
\overrightarrow{e}_n&=\frac{\partial x_i}{\partial\theta_n}=r\sin\theta_1\sin\theta_2\ldots\sin\theta_{n-1}\left\lbrace 0, 0,\ldots, 0,-\sin\theta_n, \cos\theta_n\cos\theta_{n+1},\right.\\
& \cos\theta_n\sin\theta_{n+1}\cos\theta_{n+2},\ldots, \cos\theta_n\sin\theta_{n+1}\sin\theta_{n+2}\ldots\sin\theta_{N-3}\cos\theta_{N-2},\\
&\left. \cos\theta_n\sin\theta_{n+1}\ldots\sin\theta_{N-2}\cos\phi, \cos\theta_n\sin\theta_{n+1}\ldots\sin\theta_{N-2}\sin\phi\right\rbrace,\\
&\vdots\\
\overrightarrow{e}_{N-2}&=\frac{\partial x_i}{\partial\theta_{N-2}}=r\sin\theta_1\sin\theta_2\ldots\sin\theta_{N-3}\left\lbrace 0, 0,\ldots, 0, -\sin\theta_{N-2}, \right. \\
&\left.\cos\theta_{N-2}\cos\phi, \cos\theta_{N-2}\sin\phi\right\rbrace,\\
\overrightarrow{e}_\phi&=\frac{\partial x_i}{\partial\phi}=r\sin\theta_1\sin\theta_2\ldots\sin\theta_{N-2}\left\lbrace 0,0,\ldots,0, -\sin\phi, \cos\phi\right\rbrace,
\end{aligned}
\label{eq4}
\end{equation}
showing that all base vectors are orthogonal because
\begin{equation}
\forall\, n=1,\ldots, N-2, \quad \overrightarrow{e}_r\cdot\overrightarrow{e}_n=\overrightarrow{e}_\phi\cdot\overrightarrow{e}_n=\overrightarrow{e}_r\cdot\overrightarrow{e}_\phi=0.
\label{eq5}
\end{equation}

The hyperspherical coordinates are an orthogonal curvilinear coordinate system in $N$ dimensions, and the modulus of the base vectors specify the scale factors:
\begin{align}
\forall\, i=1,\ldots,N \quad h_i\equiv\left|
\overrightarrow{e}_i
\right|=\left\lbrace \vphantom{\sin\theta_{N-2}} 1, r, r\sin\theta_1, r\sin\theta_1\sin\theta_2,\ldots, \right.\nonumber\\
\left.r\sin\theta_1\sin\theta_2\ldots\sin\theta_{i-2},\ldots,r\sin\theta_1\sin\theta_2\ldots\sin\theta_{N-2}\right\rbrace.
\label{eq6}
\end{align}
The scale factors specify the covariant
\begin{subequations}
\begin{equation}
g_{ij}=\left(h_i\right)^2\delta_{ij} \label{eq7a}
\end{equation}
and contravariant
\begin{equation}
g^{ij}=\left(h_i\right)^{-2}\delta_{ij} \label{eq7b}
\end{equation}
\end{subequations}
metric tensors where $\delta_{ij}$ is the identity matrix \cite{Sokolnikoff1951}. The determinant $g$ of the covariant metric tensor,
\begin{subequations}
\begin{align}
\left|g\right|^{1/2}&=\prod_{i=1}^N h_i\nonumber\\
&=r^{N-1}\sin^{N-2}\theta_1\sin^{N-3}\theta_2\ldots\sin^{N-i}\theta_{i-1}\ldots\sin^2\theta_{N-3}\sin\theta_{N-2},
\label{eq8a}
\end{align}
specifies the volume element in hyperspherical coordinates \cite{Sokolnikoff1951}:
\begin{equation}
\mathrm{d} V=\left|
g
\right|^{1/2}\mathrm{d} r\,\mathrm{d}\theta_1\ldots\mathrm{d}\theta_{N-2}\,\mathrm{d}\phi.
\label{eq8b}
\end{equation}
\end{subequations}

\subsection{$N$-dimensional Laplacian in hyperspherical coordinates}
\label{sec:2.3}

The scale factors can be used to write any invariant differential operator in hyperspherical coordinates, for example the Laplacian operator. 

The Laplacian is given in terms of the metric tensor by
\begin{subequations}
\begin{equation}
\nabla^2=\frac{1}{\sqrt{g}}\frac{\partial}{\partial x_i}\left(g^{ij}\sqrt{g}\frac{\partial}{\partial x_j}\right)
\label{eq11a}
\end{equation}
that simplifies to
\begin{equation}
\nabla^2=\frac{1}{\sqrt{g}}\frac{\partial}{\partial x_i}\left[\frac{1}{\left(h_i\right)^2}\sqrt{g}\frac{\partial}{\partial x_i}\right]
\label{eq11b}
\end{equation}
\end{subequations}
for orthogonal curvilinear coordinates \cite{Sokolnikoff1951} in terms of the scale factors \eqref{eq6}. Using the scale factors in hyperspherical coordinates, the successive terms are: (i) for the radius,
\begin{subequations}
\begin{equation}
\frac{1}{\sqrt{g}}\frac{\partial}{\partial r}\left[\sqrt{g}\left(h_1\right)^{-2}\frac{\partial}{\partial r}\right]=\frac{1}{r^{N-1}}\frac{\partial}{\partial r}\left(r^{N-1}\frac{\partial}{\partial r}\right)=\frac{\partial^2}{\partial r^2}+\frac{N-1}{r}\frac{\partial}{\partial r},
\label{eq12a}
\end{equation}
which coincides with the radial part of the Laplacian in polar coordinates for $N=2$ or spherical coordinates for $N=3$; (ii) for the first latitude,
\begin{equation}
\frac{1}{\sqrt{g}}\frac{\partial}{\partial \theta_1}\left[\sqrt{g}\left(h_2\right)^{-2}\frac{\partial}{\partial \theta_1}\right]=\frac{1}{r^2\sin^{N-2}\theta_1}\frac{\partial}{\partial\theta_1}\left(\sin^{N-2}\theta_1\frac{\partial}{\partial\theta_1}\right);
\label{eq12b}
\end{equation}
(iii) for the second latitude,
\begin{equation}
\frac{1}{\sqrt{g}}\frac{\partial}{\partial \theta_2}\left[\sqrt{g}\left(h_3\right)^{-2}\frac{\partial}{\partial \theta_2}\right]=\frac{1}{r^2\sin^2\theta_1\sin^{N-3}\theta_2}\frac{\partial}{\partial\theta_2}\left(\sin^{N-3}\theta_2\frac{\partial}{\partial\theta_2}\right);
\label{eq12c}
\end{equation}
(iv) for the $n$-th latitude,
\begin{align}
\forall\,n=1,\ldots,N-2: \quad \frac{1}{\sqrt{g}}\frac{\partial}{\partial \theta_n}\left[\sqrt{g}\left(h_{n+1}\right)^{-2}\frac{\partial}{\partial \theta_n}\right]\nonumber\\
=\frac{1}{r^2\sin^2\theta_1\ldots\sin^2\theta_{n-1}\sin^{N-n-1}\theta_n}\frac{\partial}{\partial\theta_n}\left(\sin^{N-n-1}\theta_n\frac{\partial}{\partial\theta_n}\right);
\label{eq12d}
\end{align}
(v) for the last or ($N-2$)-th latitude,
\begin{align}
\frac{1}{\sqrt{g}}\frac{\partial}{\partial \theta_{N-2}}\left[\sqrt{g}\left(h_{N-1}\right)^{-2}\frac{\partial}{\partial \theta_{N-2}}\right]\nonumber\\
=\frac{1}{r^2\sin^2\theta_1\ldots\sin^2\theta_{N-3}\sin\theta_{N-2}}\frac{\partial}{\partial\theta_{N-2}}\left(\sin\theta_{N-2}\frac{\partial}{\partial\theta_{N-2}}\right);
\label{eq12e}
\end{align}
(vi) for the longitude,
\begin{equation}
\frac{1}{\sqrt{g}}\frac{\partial}{\partial\phi}\left[\sqrt{g}\left(h_N\right)^{-2}\frac{\partial}{\partial\phi}\right]=\frac{1}{r^2\sin^2\theta_1\ldots\sin^2\theta_{N-2}}\frac{\partial^2}{\partial\phi^2},
\label{eq12f}
\end{equation}
\end{subequations}
which is again familiar for polar ($N=2$) and spherical ($N=3$) coordinates.

Substitution of the equations \eqref{eq12a} to \eqref{eq12f} in \eqref{eq11b} proves that the Laplacian in hyperspherical coordinates is given by
\begin{align}
\nabla^2-\frac{\partial^2}{\partial r^2}-\frac{N-1}{r}\frac{\partial}{\partial r}&=\nabla^2-\frac{1}{r^{N-1}}\frac{\partial}{\partial r}\left(r^{N-1}\frac{\partial}{\partial r}\right)\nonumber\\
&=\frac{1}{r^2\sin^{N-2}\theta_1}\frac{\partial}{\partial\theta_1}\left(\sin^{N-2}\theta_1\frac{\partial}{\partial\theta_1}\right)\nonumber\\
&+\frac{1}{r^2\sin^{2}\theta_1\sin^{N-3}\theta_2}\frac{\partial}{\partial\theta_2}\left(\sin^{N-3}\theta_2\frac{\partial}{\partial\theta_2}\right)+\ldots\nonumber\\
&+\frac{1}{r^2\sin^{2}\theta_1\ldots\sin^{2}\theta_{n-1}\sin^{N-n-1}\theta_n}\frac{\partial}{\partial\theta_n}\left(\sin^{N-n-1}\theta_n\frac{\partial}{\partial\theta_n}\right)+\ldots\nonumber\\
&+\frac{1}{r^2\sin^{2}\theta_1\ldots\sin^2\theta_{N-3}\sin\theta_{N-2}}\frac{\partial}{\partial\theta_{N-2}}\left(\sin\theta_{N-2}\frac{\partial}{\partial\theta_{N-2}}\right)\nonumber\\
&+\frac{1}{r^2\sin^2\theta_1\ldots\sin^2\theta_{N-2}}\frac{\partial^2}{\partial\phi^2}.
\label{eq9}
\end{align}
The last term on the left-hand side (l.h.s.) and the last term on the right-hand side (r.h.s.) of \eqref{eq9}, substituting $N=2$, coincide with the well-known formula
\begin{subequations}
\begin{equation}
\nabla^2-\frac{\partial^2}{\partial r^2}-\frac{1}{r}\frac{\partial}{\partial r}=\nabla^2-\frac{1}{r}\frac{\partial}{\partial r}\left(r\frac{\partial}{\partial r}\right)=\frac{1}{r^2}\frac{\partial^2}{\partial\phi^2}
\label{eq10a}
\end{equation}
for the Laplacian in polar coordinates \cite{Abramowitz1965}. Besides, substituting $N=3$, the last term on the l.h.s. and the last two terms on the r.h.s. of \eqref{eq9} coincide with the formula
\begin{align}
\nabla^2-\frac{\partial^2}{\partial r^2}-\frac{2}{r}\frac{\partial}{\partial r}&=\nabla^2-\frac{1}{r^2}\frac{\partial}{\partial r}\left(r^2\frac{\partial}{\partial r}\right)\nonumber\\
&=\frac{1}{r^2\sin\theta}\frac{\partial}{\partial\theta}\left(\sin\theta\frac{\partial}{\partial\theta}\right)+\frac{1}{r^2\sin^2\theta}\frac{\partial^2}{\partial\phi^2}.
\label{eq10b}
\end{align}
\end{subequations}
for the Laplacian in spherical coordinates \cite{Abramowitz1965}.

\section{Solution of the Laplace, Helmholtz, wave and heat equations}
\label{sec:3}

Many physical phenomena and engineering processes are described by differential equations that may be linear or non-linear. The non-linear differential equations may be linearised by considering small perturbations of a mean state, for example, water waves in a channel \cite{Green1837} or the acoustics of horns \cite{Rayleigh1916}. The linear differential equations have constant coefficients in a steady homogeneous medium using Cartesian coordinates, and otherwise have variable coefficients, that is: (i) for an inhomogeneous medium, even in one dimension; (ii) for an homogeneous medium using curvilinear coordinates. An example of (ii) is the generalized isotropic equation of mathematical physics whose the spatial dependence is specified by the Laplacian operator in hyperspherical coordinates.

The Laplacian appears in several of the most important equations of mathematical physics, such as the Laplace, Helmholtz, wave, heat, Schr\"{o}dinger, telegraph, Klein-Gordon equations and their combinations. A partial differential equation combining all these (subsection \ref{sec:3.1}) is solved by separation of variables in hyperspherical coordinates (subsection \ref{sec:3.2}) involving: (i) exponential functions for the dependencies on time and longitude; (ii) Bessel functions extending spherical Bessel functions for the dependence on the radius; (iii) the dependence on the last latitude that is specified by an associated Legendre function as in spherical harmonics; (iv) the dependence on the remaining $N-3$ latitudes that specifies an extension to hyperspherical associated Legendre functions (subsection \ref{sec:3.3}).

\subsection{A general equation of mathematical physics}
\label{sec:3.1}

The equation of mathematical physics is defined as a linear differential operator in space-time, with: (i) spatial dependence specified by the Laplace operator; (ii) time dependence specified by a linear differential operator of time with constant coefficients. The later operator (ii) is usually of the second-order ($L=2$), but can be expanded to any order, with solution by separation of variables being possible in both cases.

The original equation of mathematical physics is a linear differential equation with spatial dependence specified by the Laplace operator and temporal dependence as a linear second-order differential operator with constant coefficients:
\begin{subequations}
\begin{equation}
\nabla^2F=\frac{1}{c^2}\frac{\partial^2 F}{\partial t^2}+\frac{1}{b}\frac{\partial F}{\partial t}+aF.
\label{eq13a}
\end{equation}
The equation of mathematical physics \eqref{eq13a} includes: (i) the Laplace equation if the r.h.s. is zero; (ii) the wave equation with propagation speed $c$ for the first term on the r.h.s.; (iii) the heat equation with diffusivity $b$ for the second term on the r.h.s.; (iv) the Helmholtz equation with parameter $a$ for the third term on the r.h.s.; (v) combinations of the preceding, such as the telegraph equation consisting of the l.h.s. plus the three terms on the r.h.s. and the Klein-Gordon equation involving, besides the l.h.s., the first and third terms on the r.h.s..

The extended equation of mathematical physics
\begin{equation}
\nabla^2 F=\sum_{l=0}^L A_l\frac{\partial^l F}{\partial t^l}
\label{eq13b}
\end{equation}
\end{subequations}
is a linear partial differential equation in space-time with constant coefficients where: (i) the temporal part is a linear partial differential operator with constant coefficients of any order; (ii) the spatial part is the Laplacian that has second-order derivatives.

All these equations can be reduced to an Helmholtz equation if the Fourier transform in time is performed. This reduction is reviewed in the appendix \ref{sec:appA} and it is useful to introduce the solutions in hypercylindrical coordinates, determined in the appendix \ref{sec:appB}. The comparison between the solutions in both systems of coordinates are also detailed in the appendix \ref{sec:appB}. After the transformation, the spatial dependence specified by the Laplace operator remains constant and therefore the only difference between the equations of mathematical physics in time domain or in frequency domain is related to the temporal dependence. Consequently, the comparison of the solutions between hyperspherical and hypercylindrical coordinates can be made in different domains of time, since the system of coordinates influences only the spatial part of the solutions and the temporal part of the solutions is uniquely related to the time/frequency domains. It means that using different systems of coordinates influences only the spatial part of the solutions.

Using the Laplacian in hyperspherical coordinates \eqref{eq9}, the original equation of mathematical physics \eqref{eq13a}, also in hyperspherical coordinates, takes the form
\begin{subequations}
\begin{align}
r^2\left(aF+\frac{1}{b}\frac{\partial F}{\partial t}+\frac{1}{c^2}\frac{\partial^2 F}{\partial t^2}\right)&=r^2\nabla^2F=r^2\left(\frac{\partial^2F}{\partial r^2}+\frac{N-1}{r}\frac{\partial F}{\partial r}\right)\nonumber\\
&+\csc^{N-2}\theta_1\frac{\partial}{\partial\theta_1}\left(\sin^{N-2}\theta_1\frac{\partial F}{\partial\theta_1}\right)\nonumber\\
&+\csc^2\theta_1\left[\csc^{N-3}\theta_2\frac{\partial}{\partial\theta_2}\left(\sin^{N-3}\theta_2\frac{\partial F}{\partial\theta_2}\right)\right.\nonumber\\
&+\csc^2\theta_2\left[\csc^{N-4}\theta_3\frac{\partial}{\partial\theta_3}\left(\sin^{N-4}\theta_3\frac{\partial F}{\partial\theta_3}\right)+\ldots\right.\nonumber\\
&+\csc^2\theta_{n-1}\left[\csc^{N-n-1}\theta_n\frac{\partial}{\partial\theta_n}\left(\sin^{N-n-1}\theta_n\frac{\partial F}{\partial\theta_n}\right)+\ldots\right.\nonumber\\
&+\csc^2\theta_{N-3}\left[\csc\theta_{N-2}\frac{\partial}{\partial\theta_{N-2}}\left(\sin\theta_{N-2}\frac{\partial F}{\partial\theta_{N-2}}\right)\right.\nonumber\\
&+\csc^2\theta_{N-2}\left.\frac{\partial^2 F}{\partial\phi^2}\right]\ldots\left.\vphantom{\frac{\partial^2 F}{\partial\phi^2}}\right]\ldots\left.\left.\vphantom{\frac{\partial^2 F}{\partial\phi^2}}\right]\right].
\label{eq14a}
\end{align}
The extended equation of mathematical physics in hyperspherical coordinates takes the same form as \eqref{eq14a}, but the first member becomes
\begin{equation}
r^2\sum_{l=0}^L A_l\frac{\partial^l F}{\partial t^l}=r^2\nabla^2F.
\label{eq14b}
\end{equation}
\end{subequations}
The Laplacian operator in hyperspherical coordinates has been written in \eqref{eq13a} in a ``nested form'', taking factors out of the brackets as early as possible. This facilitates the solution by separation of variables, as shown next.

The solution of the original and extended equations of mathematical physics in hyperspherical coordinates is obtained by separation of variables
\begin{equation}
F\left(r, \theta_1, \ldots, \theta_{N-2}, \phi, t\right)=T\left(t\right)\,R\left(r\right)\,\Phi\left(\phi\right)\prod_{n=1}^{N-2}\Theta_n\left(\theta_n\right)
\label{eq15}
\end{equation}
with substitution in \eqref{eq14a} and division by $F$ leading to
\begin{subequations}
\begin{align}
r^2\left(a+\frac{1}{b}\frac{1}{T}\frac{\mathrm{d}T}{\mathrm{d}t}+\frac{1}{c^2}\frac{1}{T}\frac{\mathrm{d}^2T}{\mathrm{d}t^2}\right)&=\frac{r^2}{R}\left(\frac{\mathrm{d}^2R}{\mathrm{d}r^2}+\frac{N-1}{r}\frac{\mathrm{d}R}{\mathrm{d}r}\right)\nonumber\\
&+\frac{1}{\Theta_1}\frac{\mathrm{d}^2\Theta_1}{\mathrm{d}\theta_1^2}+\left(N-2\right)\cot\theta_1\frac{1}{\Theta_1}\frac{\mathrm{d}\Theta_1}{\mathrm{d}\theta_1}\nonumber\\
&+\csc^2\theta_1\left[\frac{1}{\Theta_2}\frac{\mathrm{d}^2\Theta_2}{\mathrm{d}\theta_2^2}+\left(N-3\right)\cot\theta_2\frac{1}{\Theta_2}\frac{\mathrm{d}\Theta_2}{\mathrm{d}\theta_2}\right.\nonumber\\
&+\csc^2\theta_2\left[\frac{1}{\Theta_3}\frac{\mathrm{d}^2\Theta_3}{\mathrm{d}\theta_3^2}+\left(N-4\right)\cot\theta_3\frac{1}{\Theta_3}\frac{\mathrm{d}\Theta_3}{\mathrm{d}\theta_3}+\ldots\right.\nonumber\\
&+\csc^2\theta_{n-1}\left[\frac{1}{\Theta_n}\frac{\mathrm{d}^2\Theta_n}{\mathrm{d}\theta_n^2}+\left(N-n-1\right)\cot\theta_n\frac{1}{\Theta_n}\frac{\mathrm{d}\Theta_n}{\mathrm{d}\theta_n}+\ldots\right.\nonumber\\
&+\csc^2\theta_{N-3}\left[\frac{1}{\Theta_{N-2}}\frac{\mathrm{d}^2\Theta_{N-2}}{\mathrm{d}\theta_{N-2}^2}+\cot\theta_{N-2}\frac{1}{\Theta_{N-2}}\frac{\mathrm{d}\Theta_{N-2}}{\mathrm{d}\theta_{N-2}}\right.\nonumber\\
&+\csc^2\theta_{N-2}\frac{1}{\Phi}\frac{\mathrm{d}^2\Phi}{\mathrm{d}\phi^2}\left.\vphantom{\frac{\mathrm{d}^2\Theta_{N-2}}{\mathrm{d}\theta_{N-2}^2}}\right]\ldots\left.\vphantom{\frac{\mathrm{d}^2\Theta_{N-2}}{\mathrm{d}\theta_{N-2}^2}}\right]\ldots\left.\left.\vphantom{\frac{\mathrm{d}^2\Theta_{N-2}}{\mathrm{d}\theta_{N-2}^2}}\right]\right]
\label{eq16}
\end{align}
separating the variables $\left(t, r, \theta_1, \theta_2, \ldots, \theta_n, \ldots, \theta_{N-2}, \phi\right)$ as much as possible. If the substitution was made in \eqref{eq14b}, the result would have been similar to \eqref{eq15}, but the l.h.s. would be
\begin{equation}
r^2\left(A_0+\sum_{l=1}^L A_l\frac{1}{T}\frac{\mathrm{d}^l T}{\mathrm{d} t^l}\right).
\label{eq16'}
\end{equation}
\end{subequations}

\subsection{Separation of variables and a set of $N+1$ ordinary differential equations}
\label{sec:3.2}

The solution of the original \eqref{eq13a} and extended \eqref{eq13b} equations of mathematical physics in hyperspherical coordinates by separation of variables, stated in \eqref{eq14a}, leads to a set of $N+1$ ordinary differential equations considered next.

The equation \eqref{eq15} leads to a set of $N+1$ ordinary differential equations one for each factor in \eqref{eq16} because: (i) the r.h.s. of \eqref{eq16} does not depend on time, so the term in curved brackets on the l.h.s. must be a constant, namely the symmetric of the radial wavenumber $k^2$, with
\begin{subequations}
\begin{equation}
\frac{1}{c^2}\frac{\mathrm{d}^2T}{\mathrm{d} t^2}+\frac{1}{b}\frac{\mathrm{d} T}{\mathrm{d} t}+\left(a+k^2\right)T=0
\label{eq17a}
\end{equation}
for the original equation of mathematical physics, and
\begin{equation}
\sum_{l=1}^L A_l\frac{\mathrm{d}^l T}{\mathrm{d} t^l}+\left(A_0+k^2\right)T=0
\label{eq17a.1}
\end{equation}
for the extended equation; (ii) the first term on the r.h.s. of \eqref{eq16} is the only one depending on the radius, so, additionally with the l.h.s., it is also a constant, and denoting that constant by $q\left(q+1\right)$ leads to
\begin{equation}
r^2\frac{\mathrm{d}^2R}{\mathrm{d} r^2}+\left(N-1\right)r\frac{\mathrm{d} R}{\mathrm{d} r}+\left[k^2r^2-q\left(q+1\right)\right]R=0;
\label{eq17b}
\end{equation}
(iii) the last factor on the r.h.s. of \eqref{eq16} is the only one depending on the longitude, so it must be a constant, equal to $-m^2$ where $m$ is the azimuthally wavenumber, leading to
\begin{equation}
\frac{\mathrm{d}^2\Phi}{\mathrm{d}\phi^2}+m^2\Phi=0;
\label{eq17c}
\end{equation}
(iv) the last latitude $\theta_{N-2}$ appears only in the last two terms on the r.h.s. of \eqref{eq16} and must be a constant leading to
\begin{align}
\frac{1}{\Theta_{N-2}}\frac{\mathrm{d}^2\Theta_{N-2}}{\mathrm{d}\theta^2_{N-2}}&+\cot\theta_{N-2}\frac{1}{\Theta_{N-2}}\frac{\mathrm{d}\Theta_{N-2}}{\mathrm{d}\theta_{N-2}}\nonumber\\
&+\csc^2\theta_{N-2}\frac{1}{\Phi}\frac{\mathrm{d}^2\Phi}{\mathrm{d}\phi^2}=-q_{N-2}\left(1+q_{N-2}\right)
\label{eq17d}
\end{align}
which on account of \eqref{eq17c} is equivalent to
\begin{align}
\frac{\mathrm{d}^2\Theta_{N-2}}{\mathrm{d}\theta^2_{N-2}}&+\cot\theta_{N-2}\frac{\mathrm{d}\Theta_{N-2}}{\mathrm{d}\theta_{N-2}}\nonumber\\
&+\left[q_{N-2}\left(1+q_{N-2}\right)-m^2\csc^2\theta_{N-2}\right]\Theta_{N-2}=0;
\label{eq17e}
\end{align}
(v) a similar reasoning of (iv) leads, for $\theta_{N-3}$, to
\begin{align}
\frac{\mathrm{d}^2\Theta_{N-3}}{\mathrm{d}\theta^2_{N-3}}&+2\cot\theta_{N-3}\frac{\mathrm{d}\Theta_{N-3}}{\mathrm{d}\theta_{N-3}}\nonumber\\
&+\left[q_{N-3}\left(1+q_{N-3}\right)-q_{N-2}\left(1+q_{N-2}\right)\csc^2\theta_{N-3}\right]\Theta_{N-3}=0;
\label{eq17f}
\end{align}
(vi) the corresponding ordinary differential equation for $\theta_{N-n}$, with $n=3,\ldots,N-1$, is
\begin{align}
\frac{\mathrm{d}^2\Theta_{N-n}}{\mathrm{d}\theta^2_{N-n}}&+\left(n-1\right)\cot\theta_{N-n}\frac{\mathrm{d}\Theta_{N-n}}{\mathrm{d}\theta_{N-n}}\nonumber\\
&+\left[q_{N-n}\left(1+q_{N-n}\right)-q_{N-n+1}\left(1+q_{N-n+1}\right)\csc^2\theta_{N-n}\right]\Theta_{N-n}=0;
\label{eq17g}
\end{align}
\end{subequations}
(vii) the constants introduced at each stage of \eqref{eq16} lead to the equality $q_1=q$. In the case of spherical coordinates, with $N=3$, there is only one latitude and therefore the equations \eqref{eq17f} and \eqref{eq17g} do not exist. The only latitude in that case is $\theta_{N-2}=\theta_1$ (equal to $\theta$ in the appendix \ref{sec:appA.3}) satisfying \eqref{eq17e} that is equivalent to \eqref{eq8.359b}.

\subsection{Dependencies on longitude, time, radius and latitudes}
\label{sec:3.3}

The solution of the preceding set of $N+1$ ordinary differential equations involves 3 known functions and a new differential equation.

The simplest ordinary differential equation \eqref{eq17c} is for longitude, and specifies two sinusoids for $0\leq\phi\leq2\pi$ and for all non-negative integers $m$:
\begin{equation}
\Phi\left(\phi\right)=\exp\left(\pm\mathrm{i} m\phi\right).
\label{eq18}
\end{equation}
The dependence on time, stated in \eqref{eq17a} or \eqref{eq17a.1}, are two exponential functions,
\begin{subequations}
\begin{equation}
T\left(t\right)=\exp\left(\pm\mathrm{i}\omega t\right),
\label{eq19a}
\end{equation}
with frequency $\omega$ satisfying the dispersion relation
\begin{equation}
k^2=\left(\frac{\omega}{c}\right)^2-a\mp\frac{\mathrm{i}\omega}{b}
\label{eq19b}
\end{equation}
for the original equation of mathematical physics, and depending on the sign of the exponential in \eqref{eq19a}, or
\begin{equation}
k^2=-A_0-\sum_{l=1}^L A_l\left(\pm\mathrm{i}\omega\right)^l
\label{eq19b'}
\end{equation}
\end{subequations}
for the extended equation, also depending on the sign in \eqref{eq19a}. It involves the radial wavenumber $k$, which also appears in the radial dependence \eqref{eq17b}. 

The radial dependence is a cylindrical Bessel equation of order $\sigma^2=q\left(q+1\right)$ for $N=2$. Otherwise, the radial dependence is specified by a spherical Bessel equation \eqref{eq8.361b} for $N=3$, that is reducible to a cylindrical form \eqref{eq8.353b} via a change of dependent variable that involves multiplication by a factor $1/\sqrt{r}$. The factor can be interpreted as $r^{1-N/2}$ when $N=3$. This suggests the change of dependent variable
\begin{subequations}
\begin{equation}
R\left(r\right)=r^{1-N/2}S\left(r\right)
\label{eq20a}
\end{equation}
in the $N$-dimensional case that transforms to a Bessel equation
\begin{equation}
r^2S''+rS'+\left[k^2r^2-\left(\frac{N}{2}-1\right)^2-q\left(q+1\right)\right]S=0
\label{eq20b}
\end{equation}
\end{subequations}
of order $\sigma$ specified by
\begin{subequations}
\begin{equation}
\sigma^2=q\left(q+1\right)+\left(\frac{N}{2}-1\right)^2;
\label{eq21a}
\end{equation}
thus the radial dependence is specified by
\begin{equation}
R\left(r\right)=r^{1-N/2}Z_\sigma\left(kr\right)
\label{eq21b}
\end{equation}
\end{subequations}
in terms of a linear combination of Bessel, Neumann or Hankel functions:
\begin{equation}
Z_\sigma\left(kr\right)\equiv A_1J_\sigma\left(kr\right)+A_2Y_\sigma\left(kr\right)=A_+H_\sigma^{\left(1\right)}\left(kr\right)+A_-H_\sigma^{\left(2\right)}\left(kr\right).
\label{eq22}
\end{equation}
In two dimensions, $N=2$, the solution is a cylindrical Bessel function of order $\sigma^2=q\left(q+1\right)$. In three dimensions, $N=3$, from \eqref{eq21a} follows $\sigma=q+1/2$ leading to spherical Bessel functions $r^{-1/2}Z_{q+1/2}\left(kr\right)$ in \eqref{eq21b}. The agreement of \eqref{eq8.362} with \eqref{eq22}, for $N=3$,  shows that the spherical \eqref{eq8.362} and cylindrical \eqref{eq22} Bessel and Neumann functions are related by
\begin{subequations}
\begin{align}
j_q\left(kr\right)&=\sqrt{\frac{\pi}{2r}}J_{q+1/2}\left(kr\right), \label{eq8.387a}\\
y_q\left(kr\right)&=\sqrt{\frac{\pi}{2r}}Y_{q+1/2}\left(kr\right),
\label{eq8.387b}
\end{align}
\end{subequations}
where the constant factor $\sqrt{\pi/2}$ was inserted for agreement with the literature and can be absorbed into the arbitrary constants $A_+$ and $A_-$. For all higher dimensions, $N=4, 5,\ldots$, then \eqref{eq21a} specifies the order of the Bessel function \eqref{eq22} in the radial solution \eqref{eq21b}. The radial dependence of the solution for $N$ dimensions in terms of Bessel and Neumann functions \eqref{eq21b} simplifies in two dimensions to cylindrical Bessel and Neumann functions \eqref{eq8.354b} of order $\sqrt{q\left(q+1\right)}$, but variable $kr$, while in three dimensions the solution reduces to spherical Bessel and Neumann functions \eqref{eq8.362} of order $q+1/2$.

The last latitude \eqref{eq17e} satisfies an associated Legendre equation,
\begin{equation}
\Theta_{N-2}\left(\theta_{N-2}\right)=P^m_{q_{N-2}}\left(\cos\theta_{N-2}\right),
\label{eq23}
\end{equation}
as for spherical harmonics, stated in \eqref{eq8.359b}. In the case $N=3$, then $q_{N-2}=q_1=q$ in \eqref{eq17b}, leading to associated Legendre polynomials equivalent to \eqref{eq8.360}. In higher dimensions, $N=4, 5, \ldots$, besides $q_{N-2}$, there are the constants $q_{N-3}, \ldots, q_1$ appearing in \eqref{eq17g}, which lead to a generalized or hyperspherical associated Legendre function
\begin{subequations}
\begin{equation}
G\left(\theta\right)\equiv P^\mu_{\nu, \lambda}\left(\cos\theta\right)
\label{eq24a}
\end{equation}
defined by the solution of the corresponding differential equation. 

The hyperspherical associated Legendre functions are therefore defined as solutions of the differential equation
\begin{equation}
\frac{\mathrm{d}^2G}{\mathrm{d}\theta^2}+\left(1+2\lambda\right)\cot\theta\frac{\mathrm{d} G}{\mathrm{d}\theta}+\left[\nu\left(\nu+1\right)-\mu^2\csc^2\theta\right]G\left(\theta\right)=0.
\label{eq24b}
\end{equation}
\end{subequations}
The original associated Legendre differential equation is the particular case $\lambda=0$ of \eqref{eq24b}. The dependence on the ($N-n$)-th latitude, according to \eqref{eq17g}, that is, with $n=3, \ldots, N-1$, is of the type \eqref{eq24b}, which is specified by an hyperspherical associated Legendre function
\begin{subequations}
\begin{equation}
\Theta_{N-n}\left(\theta_{N-n}\right)=P^{\mu_n}_{\nu_{n}, \lambda_{n}}\left(\cos\theta_{N-n}\right)
\label{eq25d}
\end{equation}
involving the parameters
\begin{align}
\nu_{n}&=q_{N-n}, \label{eq25a}\\
\mu_n&=\left|
q_{N-n+1}\left(1+q_{N-n+1}\right)
\right|^{1/2}, \label{eq25b}\\
2\lambda_{n}&=n-2. \label{eq25c}
\end{align}
\end{subequations}

\section{General solution in terms of hyperspherical associated Legendre functions}
\label{sec:4}

The solution of the general equation of mathematical physics is summarized next in $N$ dimensions (subsection \ref{sec:4.1}), and reviewed in the four-dimensional case (subsection \ref{sec:4.2}). The latter involves the first of the hyperspherical associated Legendre functions, which can be defined for any dimension (subsection \ref{sec:4.3}).

\subsection{Solution of the general equation of mathematical physics}
\label{sec:4.1}

Remembering the equations \eqref{eq15} and \eqref{eq17b} to \eqref{eq17g}, the solution of the original or extended equation of mathematical physics in hyperspherical coordinates \eqref{eq14b}, using the method separation of variables, is given by
\begin{align}
F\left(r, \theta_1, \ldots, \theta_{N-2}, \phi, t\right)&=r^{1-N/2}\sum_{m=-\infty}^{+\infty}\mathrm{e}^{\mathrm{i}\left(m\phi-\omega t\right)}\sum_{q_{1}=1}^\infty Z_\sigma\left(kr\right)\nonumber\\
&\times\sum_{q_{N-2}=1}^\infty P_{q_{N-2}}^m\left(\cos\theta_{N-2}\right)\sum_{q_{N-3},\ldots,q_{1}=1}^\infty\prod_{n=3}^{N-1}P_{q_{N-n}, \lambda_{n}}^{\mu_{n}}\left(\cos\theta_{N-n}\right)
\label{eq26}
\end{align}
as a product of the following factors: (i) the radial dependence with amplitude
\begin{subequations}
\begin{equation}
A\sim r^{1-N/2}
\label{eq27a}
\end{equation}
leads to a flux
\begin{equation}
A^2r^{N-1}\sim\mathrm{const.}
\label{eq27b}
\end{equation}
\end{subequations}
from \eqref{eq21b} through an hypersphere of radius $R$ which is independent of the radius, hence a constant; (ii) the time dependence \eqref{eq19a} involves a frequency $\omega$, which is a root of the dispersion relation \eqref{eq19b} or \eqref{eq19b'} for original or extended equation respectively, and hence may be complex, allowing not only sinusoidal oscillations if $\Re\left(\omega\right)\neq0$, but also decays if $\Im\left(\omega\right)<0$ or grows if $\Im\left(\omega\right)>0$ with time; (iii) the dependence on longitude \eqref{eq18} is sinusoidal with integer wavenumber $m$; (iv) the radial dependence involves besides the constant flux factor \eqref{eq27a}, the Bessel functions \eqref{eq22} of order $\sigma$ satisfying \eqref{eq21a} and involving $q=q_1$; (v) the latter is the first of fundamental wavenumbers $\left(q_1,\ldots,q_{N-2}\right)$ appearing in \eqref{eq17c} to \eqref{eq17g}; (vi) the last latitude $\theta_{N-2}$ is specified by an associated Legendre function \eqref{eq23}, so the wavenumber $q_{N-2}$ is a positive integer if the direction $\theta_{N-2}=\pi$ is included; (vii) the remaining $N-3$ latitudes each appear as a factor specified by hyperspherical associated Legendre functions \eqref{eq25d}, with \eqref{eq25a} as the degree, \eqref{eq25b} as the order and \eqref{eq25c} as the dimension; (vii) the first latitude $\theta_1$ corresponds to the order $q_1=q$ and to the hyperspherical associated Legendre function of variable $\theta_1$, degree $q$, order $\mu_{N-1}=\left|
q_2\left(1+q_2\right)
\right|^{1/2}$ and dimension $\left(N-3\right)/2$. The solution in hyperspherical coordinates of the Helmholtz equation, using the Fourier transform in time, and its comparison with the hypercylindrical coordinates can be reviewed in the appendix \ref{sec:appB.4}.

The four-dimensional case is reviewed in the next subsection before proceeding to consider the main new feature, namely the hyperspherical associated Legendre functions. The four-dimensional case in hypercylindrical coordinates can be reviewed in the appendix \ref{sec:appB.3}.

\subsection{Particular case of four-dimensional harmonics}
\label{sec:4.2}

The simplest case beyond spherical harmonics is four-dimensional and is reviewed combining all preceding results in this subsection. 

The four-dimensional hyperspherical coordinates $\left(r, \psi, \theta, \phi\right)$ are related to the Cartesian coordinates $\left(x_1, x_2, x_3, x_4\right)\equiv\left(z, w, x, y\right)$ through the equations \eqref{eq2} leading to
\begin{align}
0\leq r<\infty, \quad 0\leq\psi, \theta\leq\pi, \quad 0\leq\phi\leq2\pi: \nonumber\\ \left\lbrace z, w, x, y\right\rbrace=r\left\lbrace\cos\psi, \sin\psi\cos\theta, \sin\psi\sin\theta\cos\phi, \sin\psi\sin\theta\sin\phi\right\rbrace. \label{eq28}
\end{align}
The inverse coordinate transformation from four-dimension Cartesian to hyperspherical coordinates, using \eqref{eq3}, is
\begin{equation}
\begin{aligned}
r&=\left|z^2+w^2+x^2+y^2\right|^{1/2},\\
\cot\psi&=z\left|
w^2+x^2+y^2
\right|^{-1/2},\\
\cot\theta&=w\left|x^2+y^2\right|^{-1/2},\\
\cot\phi&=x/y.
\end{aligned}
\label{eq29}
\end{equation}
The hyperspherical base vectors, defined in \eqref{eq4}, are
\begin{equation}
\begin{aligned}
\overrightarrow{e}_r&=\left\lbrace\cos\psi, \sin\psi\cos\theta, \sin\psi\sin\theta\cos\phi, \sin\psi\sin\theta\sin\phi\right\rbrace,\\
\overrightarrow{e}_\psi&=r\left\lbrace-\sin\psi, \cos\psi\cos\theta, \cos\psi\sin\theta\cos\phi, \cos\psi\sin\theta\sin\phi\right\rbrace,\\
\overrightarrow{e}_\theta&=r\sin\psi\left\lbrace 0, -\sin\theta, \cos\theta\cos\phi, \cos\theta\sin\phi\right\rbrace,\\
\overrightarrow{e}_\phi&=r\sin\psi\sin\theta\left\lbrace 0, 0, -\sin\phi, \cos\phi\right\rbrace.
\end{aligned}
\label{eq30}
\end{equation}
They are mutually orthogonal and specify the scale factors \eqref{eq6} resulting in
\begin{subequations}
\begin{equation}
\begin{aligned}
h_r=1, \quad h_\psi=r, \quad h_\theta=r\sin\psi, \quad h_\phi=r\sin\psi\sin\theta,
\end{aligned}
\label{eq31a-d}
\end{equation}
as well as the determinant of the covariant metric tensor 
\begin{equation}
\left|g\right|^{1/2}=r^3\sin^2\psi\sin\theta
\label{eq31e}
\end{equation}
and the volume element
\begin{equation}
\mathrm{d} V=\left|g\right|^{1/2}\mathrm{d} r\,\mathrm{d}\psi\,\mathrm{d}\theta\,\mathrm{d}\phi.
\label{eq31f}
\end{equation}
\end{subequations}
The four-dimensional Laplacian in hyperspherical coordinates, knowing \eqref{eq9}, is
\begin{align}
\nabla^2&=\frac{1}{r^3}\frac{\partial}{\partial r}\left(r^3\frac{\partial}{\partial r}\right)+\frac{1}{r^2\sin^2\psi}\frac{\partial}{\partial\psi}\left(\sin^2\psi\frac{\partial}{\partial\psi}\right)\nonumber\\
&+\frac{1}{r^2\sin^2\psi\sin\theta}\frac{\partial}{\partial\theta}\left(\sin\theta\frac{\partial}{\partial\theta}\right)+\frac{1}{r^2\sin^2\psi\sin^2\theta}\frac{\partial^2}{\partial\phi^2}.
\label{eq32}
\end{align}
The original and general equations of mathematical physics, \eqref{eq13a} and \eqref{eq13b} respectively, have the solution by separation of variables \eqref{eq15},
\begin{equation}
F\left(r, \psi, \theta, \phi, t\right)=T\left(t\right)\,R\left(r\right)\,\Psi\left(\psi\right)\,\Theta\left(\theta\right)\,\Phi\left(\phi\right),
\label{eq33}
\end{equation} 
where
\begin{align}
r^2\left(a+\frac{1}{b}\frac{1}{T}\frac{\mathrm{d} T}{\mathrm{d} t}+\frac{1}{c^2}\frac{1}{T}\frac{\mathrm{d}^2 T}{\mathrm{d} t^2}\right)&=\frac{r^2}{R}\left(\frac{\mathrm{d}^2 R}{\mathrm{d} r^2}+\frac{3}{r}\frac{\mathrm{d} R}{\mathrm{d} r}\right)+\frac{1}{\Psi}\left(\frac{\mathrm{d}^2\Psi}{\mathrm{d}\psi^2}+2\cot\psi\frac{\mathrm{d}\Psi}{\mathrm{d}\psi}\right)\nonumber\\
&+\csc^2\psi\left[\frac{1}{\Theta}\left(\frac{\mathrm{d}^2\Theta}{\mathrm{d}\theta^2}+\cot\theta\frac{\mathrm{d}\Theta}{\mathrm{d}\theta}\right)+\csc^2\theta\frac{1}{\Phi}\frac{\mathrm{d}^2\Phi}{\mathrm{d}\phi^2}\right]
\label{eq34}
\end{align}
with respect to the original equation. This leads to a set of five ordinary differential equations specifying the dependence, in particular, on the time through \eqref{eq19a} and on the longitude through \eqref{eq18}. The dependence on the radius is specified by \eqref{eq17b}, in this case by
\begin{subequations}
\begin{equation}
r^2\frac{\mathrm{d}^2 R}{\mathrm{d} r^2}+3r\frac{\mathrm{d} R}{\mathrm{d} r}+\left[k^2r^2-q\left(q+1\right)\right]R=0
\label{eq35a}
\end{equation}
whose solution is specified by Bessel functions \eqref{eq21b} and \eqref{eq22},
\begin{equation}
R\left(r\right)=\frac{1}{r}Z_\sigma\left(kr\right),
\label{eq35b}
\end{equation}
of order $\sigma$ with
\begin{equation}
\sigma^2=q^2+q+1.
\label{eq35c}
\end{equation}
\end{subequations}
The dependence on the last latitude, regarding \eqref{eq23}, with $s\equiv q_2$, is 
\begin{subequations}
\begin{equation}
\Theta\left(\theta\right)=P_s^m\left(\cos\theta\right)
\label{eq36a}
\end{equation}
satisfying the differential equation
\begin{equation}
\frac{\mathrm{d}^2\Theta}{\mathrm{d}\theta^2}+\cot\theta\frac{\mathrm{d}\Theta}{\mathrm{d}\theta}+\left[s\left(s+1\right)-m^2\csc^2\theta\right]\Theta=0.
\label{eq36b}
\end{equation}
\end{subequations}
Lastly, the dependence on the first latitude, substituting $n=N-1$ in \eqref{eq25d} to \eqref{eq25c}, with $q=q_1$, satisfies
\begin{subequations}
\begin{equation}
\frac{\mathrm{d}^2\Psi}{\mathrm{d}\psi^2}+2\cot\psi\frac{\mathrm{d}\Psi}{\mathrm{d}\psi}+\left[q\left(q+1\right)-s\left(s+1\right)\csc^2\psi\right]\Psi=0
\label{eq37a}
\end{equation}
whose solution is
\begin{equation}
\Psi\left(\psi\right)=P_{q, 1/2}^{\mu_{N-1}}\left(\cos\psi\right)
\label{eq37c}
\end{equation}
with
\begin{equation}
\mu_{N-1}=\left|
s\left(s+1\right)
\right|^{1/2}.
\label{eq37b}
\end{equation}
\end{subequations}
The general solution is
\begin{align}
F\left(r, \psi, \theta, \phi, t\right)&=\frac{1}{r}\mathrm{e}^{-\mathrm{i}\omega t}\sum_{m=-\infty}^{+\infty}\mathrm{e}^{\mathrm{i} m\phi} \sum_{q=1}^{\infty} Z_\sigma\left(kr\right)\nonumber\\
&\times\sum_{s=1}^\infty P_s^m\left(\cos\theta\right)\sum_{q=1}^\infty P_{q, 1/2}^{\mu_{N-1}}\left(\cos\psi\right)
\label{eq38}
\end{align}
where the Bessel functions have order \eqref{eq35c}, the original associated Legendre polynomials specify the dependence on the second latitude $\theta$ and the dependence on the first latitude $\psi$ is specified by hyperspherical associated Legendre functions of dimension $1/2$ in \eqref{eq37c} with $\mu_{N-1}$ given by \eqref{eq37b}. For higher dimensions, $N=5, 6\ldots$, more hyperspherical associated Legendre functions appear.

The first latitude that corresponds to the hyperspherical associated Legendre functions is given by the product of the sine function and the hypergeoemtric function. The explicit relation is evaluated in the theorem \ref{thm3}. The function, from \eqref{eq25d}, depends on the parameters $\mu$, $\nu$ and $\lambda$. For four dimensions, $N$ is equal to $4$ and the first latitude corresponds to $n=3$. Therefore, there are two values that influence the hyperspherical associated Legendre function given by \eqref{eq37c} (the value of $\lambda$ is equal to $1/2$): $q_1=q$ and $q_2=s$. These two values, in that order, are indicated between curved parentheses at each plot of the figure \ref{fig0}. Assigning integer values, the figure \ref{fig0} shows several plots of the hyperspherical associated Legendre function for values of the first latitude $\psi$ between $0$ and $\pi$. Each plot represents the two independent solutions. The solid line is the first solution with the plus sign in the first equation of \eqref{eq56a-c} while the dashed line is the second solution with the negative sign. In some cases, such as $\left(q_1, q_2\right)=\left(1, 1\right)$ or $\left(q_1, q_2\right)=\left(2, 3\right)$ as sketched in the first row of plots in the figure \ref{fig0}, one parameter of the hypergeometric function ($\alpha$ or $\beta$ in \eqref{eq55b}) is zero and consequently the hypergeometric function reduces to unity and the function $\Psi$ is equal to the first term of \eqref{eq55b}. There are also cases such as the second line of plots in the figure \ref{fig0} corresponding to irrational values of the parameters $\alpha$ and $\beta$. Lastly, the third line of plots represents the functions $\Psi$ reduced to sine functions multiplied by polynomials because, in that cases, one upper parameter ($\alpha$ or $\beta$) of the hypergeometric function is a non-positive integer. All these plots show more three properties: they are symmetric with respect to $\psi=\pi/2$, they are equal to $1$ for $\psi=\pi/2$ and the functions $\Psi$ converge for $0<\theta<\pi$ and diverge for all other values, including the points $\psi=0$ and $\psi=\pi$, since the hypergeometric functions have two upper parameters ($\alpha$ and $\beta$) and one lower parameter ($\lambda$). 

\begin{figure}
\centering
\includegraphics[scale=0.9]{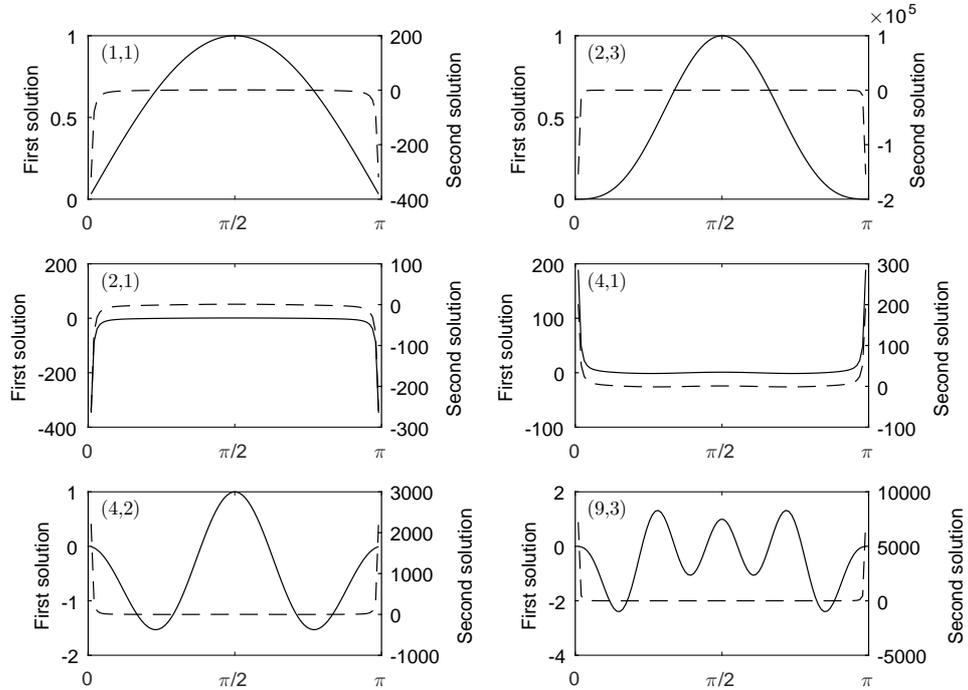}
\caption{Function $\Psi$ of the first latitude, between $0$ and $\pi$, for several integer values of $q$ and $s$.}
\label{fig0}
\end{figure}

For higher dimensions, there are more hyperspherical associated Legendre functions corresponding to additional latitudes. They differ from the case of four dimensions because the values of $\lambda$ are not equal to $1/2$, except for the latitude $\theta_{N-3}$. Nonetheless, the plot of these functions are similar to the ones of figure \ref{fig0}, but they diverge more quickly near the values $\theta=0$ and $\theta=\pi$.

The hyperspherical associated Legendre functions include as particular cases not only to the associated Legendre functions, but also to the hyperspherical Legendre functions considered in the next subsection.

\subsection{Hierarchy of related Legendre functions}
\label{sec:4.3}

The new feature of the solution of the original or general equations of mathematical physics in terms of hyperspherical coordinates is the appearance of the hyperspherical associated Legendre functions that can be related to the Gaussian hypergeometric functions (see theorem \ref{thm3}). This preceded by the separate consideration of hyperspherical Legendre functions (see lemma \ref{lem8}) and associated Legendre functions (see lemma \ref{lem9}). The associated Legendre functions are the particular case $\lambda=0$ of \eqref{eq24a} and \eqref{eq24b}, that is,
\begin{subequations} 
\begin{align}
X\left(\theta\right)\equiv P_{\nu, 0}^\mu\left(\cos\theta\right)=P_\nu^\mu\left(\cos\theta\right): \label{eq39a}\\
\frac{\mathrm{d}^2X}{\mathrm{d}\theta^2}+\cot\theta\frac{\mathrm{d} X}{\mathrm{d}\theta}+\left[\nu\left(\nu+1\right)-\mu^2\csc^2\theta\right]X=0. \label{eq39c}
\end{align}
\end{subequations}
Otherwise, the hyperspherical Legendre functions are the particular case $\mu=0$ of \eqref{eq24a} and \eqref{eq24b}, that is,
\begin{subequations} 
\begin{align}
Y\left(\theta\right)\equiv P_{\nu, \lambda}^0\left(\cos\theta\right)=P_{\nu, \lambda}\left(\cos\theta\right): \label{eq40a}\\
\frac{\mathrm{d}^2Y}{\mathrm{d}\theta^2}+\left(1+2\lambda\right)\cot\theta\frac{\mathrm{d} Y}{\mathrm{d}\theta}+\nu\left(\nu+1\right)Y=0. \label{eq40b}
\end{align}
\end{subequations}
The Legendre functions are the common particular case of \eqref{eq39a} and \eqref{eq40a}, setting $\lambda=0$ and $\mu=0$, that leads to
\begin{subequations} 
\begin{align}
Z\left(\theta\right)\equiv P_{\nu, 0}^0\left(\cos\theta\right)=P_{\nu}\left(\cos\theta\right): \label{eq41a}\\
\frac{\mathrm{d}^2Z}{\mathrm{d}\theta^2}+\cot\theta\frac{\mathrm{d} Z}{\mathrm{d}\theta}+\nu\left(\nu+1\right)Z=0. \label{eq41b}
\end{align}
\end{subequations}

Before establishing the relation with the Gaussian hypergeometric function (in section \ref{sec:5}), it is convenient to write the differential equation \eqref{eq24b} in an alternative form. The change of independent variable
\begin{subequations}
\begin{equation}
x=\cos\theta, \label{eq42a}
\end{equation}
assuming that
\begin{equation}
H\left(x\right)\equiv P_{\nu, \lambda}^\mu\left(\cos\theta\right), \label{eq42b}
\end{equation}
transforms \eqref{eq24b} to the hyperspherical associated Legendre differential equation
\begin{equation}\label{eq42c}
\left(1-x^2\right)H''-2\left(1+\lambda\right)x H'+\left[\nu\left(\nu+1\right)-\frac{\mu^2}{1-x^2}\right]H=0
\end{equation}
\end{subequations}
with variable $x$, degree $\nu$, order $\mu$ and dimension $\lambda$. The transformation of the differential equation \eqref{eq42c} to a Gaussian hypergeometric type will establish the relation with the Legendre functions of all four types, namely \eqref{eq24a}, \eqref{eq39a}, \eqref{eq40a} and \eqref{eq41a}.

\section{Relation with Gaussian hypergeometric functions}
\label{sec:5}

The relation with Gaussian hypergeometric functions is: (i) obtained first (see lemma \ref{lem8}) for hyperspherical Legendre functions (subsection \ref{sec:5.1}); (ii) reviewed next (see lemma \ref{lem9}) for associated Legendre functions (subsection \ref{sec:5.2}); (iii) finally generalized (see theorem \ref{thm3}) for hyperspherical associated Legendre functions (subsection \ref{sec:5.3}).

\subsection{Hyperspherical Legendre functions}
\label{sec:5.1}

\begin{lemma} \label{lem8}
The hyperspherical Legendre functions \eqref{eq40a} are related to the Gaussian hypergeometric functions by
\begin{subequations}
\begin{equation}
P_{\nu, \lambda}\left(\cos\theta\right)=F\left(\alpha, \beta; 1+\lambda; \frac{1-\cos\theta}{2}\right)
\label{eq43a}
\end{equation}
with parameters
\begin{equation}
\alpha, \beta=\lambda+\frac{1}{2}\pm\left|
\left(\lambda+\frac{1}{2}\right)^2+\nu\left(\nu+1\right)
\right|^{1/2}.
\label{eq43b}
\end{equation}
\end{subequations}
\end{lemma}

\begin{proof}
In the theory of the associated Legendre differential equation, the change of dependent variable
\begin{subequations}
\begin{equation}
y=\frac{1-x}{2},
\label{eq44a}
\end{equation}
defining
\begin{equation}
Q\left(y\right)\equiv H\left(x\right),
\label{eq44b}
\end{equation}
leads from \eqref{eq42c} to the ordinary differential equation
\begin{equation}
y\left(1-y\right)Q''+\left[1+\lambda-2\left(1+\lambda\right)y\right]Q'+\left[\nu\left(\nu+1\right)-\frac{\left(\mu/2\right)^2}{y\left(1-y\right)}\right]Q=0.
\label{eq44c}
\end{equation}
\end{subequations}
This last equation is of the hypergeometric type, that is, its solution is
\begin{subequations}
\begin{equation}
Q\left(y\right)=F\left(\alpha, \beta; \gamma; y\right)
\label{eq45a}
\end{equation}
because it satisfies the differential equation in the form
\begin{equation}
y\left(1-y\right)Q''+\left[\gamma-\left(\alpha+\beta+1\right)y\right]Q'-\alpha\beta Q=0
\label{eq45b}
\end{equation}
\end{subequations}
for the case $\mu=0$ with parameters satisfying
\begin{equation}
\alpha+\beta=1+2\lambda, \quad \alpha\beta=-\nu\left(\nu+1\right), \quad \gamma=1+\lambda.
\label{eq46b-d}
\end{equation}
The condition $\mu=0$ in \eqref{eq42b}, and using \eqref{eq44b} and \eqref{eq45a}, leads to
\begin{equation}
P_{\nu, \lambda}^0\left(\cos\theta\right)=Q\left(\frac{1-\cos\theta}{2}\right)=F\left(\alpha, \beta; \gamma; \frac{1-\cos\theta}{2}\right) 
\label{eq47}
\end{equation}
that proves \eqref{eq43a}. The parameters $\alpha, \beta\equiv\chi$ satisfying \eqref{eq46b-d} are the roots of
\begin{equation}
0=\left(\chi-\alpha\right)\left(\chi-\beta\right)=\chi^2-\left(\alpha+\beta\right)\chi+\alpha\beta=\chi^2-\left(1+2\lambda\right)\chi-\nu\left(\nu+1\right)
\label{eq48b}
\end{equation}
and the roots of the last equation are \eqref{eq43b}.
\end{proof}

In general, $P_{\nu, \lambda}^0\left(x\right)$ would be a linear combination of the two Gaussian hypergeometric functions which are solutions of \eqref{eq40b}. Moreover, in the case $\lambda=0$ of the Legendre functions \eqref{eq41a},
\begin{equation}
P_\nu\left(\cos\theta\right)=F\left(-\nu, 1+\nu; 1; \frac{1-\cos\theta}{2}\right).
\label{eq49}
\end{equation}
This agrees with \eqref{eq43a} and \eqref{eq43b} for $\lambda=0$, and is a known result \cite{Abramowitz1965}. Thus, \eqref{eq43a} is a definition of hyperspherical Legendre function of first kind which is consistent with the original Legendre function of first kind.

\subsection{Associated Legendre functions}
\label{sec:5.2}

Not only the hyperspherical Legendre functions (see lemma \ref{lem8}), but also the associated Legendre functions (see lemma \ref{lem9}) are related to Gaussian hypergeometric functions.

\begin{lemma} \label{lem9}
The associated Legendre functions are related to the Gaussian hypergeometric functions by
\begin{equation}
P_\nu^\mu\left(\cos\theta\right)=\frac{\left[\left(\cos\theta+1\right)/\left(\cos\theta-1\right)\right]^{\mu/2}}{\Gamma\left(1-\mu\right)}F\left(-\nu, \nu+1; 1-\mu; \frac{1-\cos\theta}{2}\right).
\label{eq50}
\end{equation}
\end{lemma}

\begin{proof}
The change of dependent variable
\begin{subequations}
\begin{equation}
H\left(x\right)=\left(\frac{x+1}{x-1}\right)^{\mu/2}J\left(x\right)
\label{eq51a}
\end{equation}
transforms \eqref{eq42c} to the ordinary differential equation
\begin{equation}
\left(1-x^2\right)J''+2\left[\mu-\left(1+\lambda\right)x\right]J'+\left[\nu\left(\nu+1\right)-\frac{2\lambda\mu x}{1-x^2}\right]J=0.
\label{eq51b}
\end{equation}
\end{subequations}
The change of independent variable
\begin{subequations}
\begin{equation}
y=\frac{1-x}{2}
\label{eq52a}
\end{equation}
with
\begin{equation}
Q\left(y\right)\equiv J\left(x\right),
\label{eq52b}
\end{equation}
similar to \eqref{eq44a} and \eqref{eq44b}, leads to
\begin{equation}
y\left(1-y\right)Q''+\left[1+\lambda-\mu-2\left(1+\lambda\right)y\right]Q'+\left[\nu\left(\nu+1\right)-\frac{\lambda\mu}{2}\frac{1/y-2}{1-y}\right]Q=0.
\label{eq52c}
\end{equation}
\end{subequations}
In the case $\lambda=0$, this is a Gaussian hypergeometric equation \eqref{eq45b} with parameters
\begin{subequations}
\begin{equation}
\gamma=1-\mu, \quad \alpha+\beta=1, \quad \alpha\beta=-\nu\left(\nu+1\right) \label{eq53b-d}
\end{equation}
implying
\begin{equation}
\alpha=-\nu, \quad \beta=1+\nu. \label{eq53e-f}
\end{equation}
\end{subequations}
Substituting \eqref{eq53e-f} in \eqref{eq45a} and using the equations \eqref{eq42a}, \eqref{eq42b}, \eqref{eq51a}, \eqref{eq52a} and \eqref{eq52b}, leads to
\begin{subequations}
\begin{equation}
P_\nu^\mu\left(\cos\theta\right)=C\left(\frac{\cos\theta+1}{\cos\theta-1}\right)^{\mu/2} F\left(-\nu, \nu+1; 1-\mu; \frac{1-\cos\theta}{2}\right)
\label{eq54a}
\end{equation}
where $C$ is an arbitrary constant. The choice of the constant
\begin{equation}
C=\frac{1}{\Gamma\left(1-\mu\right)}
\label{eq54b}
\end{equation}
\end{subequations}
in \eqref{eq54a} leads to \eqref{eq50}.
\end{proof}
The arbitrary constant in \eqref{eq54a} was chosen to agree with the known relations between associated Legendre and Gaussian hypergeometric functions \cite{Abramowitz1965}.

In the case $\lambda\neq 0$, then \eqref{eq52c} does not reduce to a Gaussian hypergeometric equation. In conclusion, it has been shown (see lemma \ref{lem8}) that the hyperspherical Legendre functions \eqref{eq40a} are a particular case of the Gaussian hypergeometric function \eqref{eq43a}. This result does not extend to the hyperspherical associated Legendre functions with $\mu\neq 0$ via the usual methods of the theory of associated Legendre functions (see lemma \ref{lem9}). By following an alternative approach, it is shown in the next subsection (see theorem \ref{thm3}) that the hyperspherical associated Legendre equation can be reduced to the Gaussian hypergeometric equation.

\subsection{Reduction to the Gaussian hypergeometric equation}
\label{sec:5.3}

Rather than follow the usual theory of associated Legendre functions, a different set of changes of independent variables is used to reduce the hyperspherical associated Legendre equation \eqref{eq42c} to the Gaussian hypergeometric type. The lemmas \ref{lem8} and \ref{lem9} are particular cases of the next theorem.

\begin{theorem} \label{thm3}
The hyperspherical associated Legendre functions are related to the Gaussian hypergeometric functions by
\begin{subequations}
\begin{equation}
P_{\nu, \lambda}^\mu\left(x\right)=\left(1-x^2\right)^\vartheta F\left(\alpha, \beta; \frac{1}{2}; x^2\right)
\label{eq55a}
\end{equation}
or alternatively
\begin{equation}
P_{\nu, \lambda}^\mu\left(\cos\theta\right)=\sin^{2\vartheta}\theta F\left(\alpha, \beta; \frac{1}{2}; \cos^2\theta\right)
\label{eq55b}
\end{equation}
\end{subequations}
where the parameters are given by
\begin{equation}
\begin{aligned}
\vartheta&=-\frac{\lambda}{2}\pm\left|
\frac{\lambda^2+\mu^2}{4}
\right|^{1/2},\\
2\left\lbrace\alpha, \beta\right\rbrace&=2\vartheta+\lambda+\frac{1}{2}\pm\left|
4\vartheta^2+\left(\lambda+\frac{1}{2}\right)^2+\nu\left(\nu+1\right)
\right|^{1/2}.
\end{aligned}
\label{eq56a-c}
\end{equation}
\end{theorem}

\begin{proof}
The change of independent variable
\begin{subequations}
\begin{equation}
z=x^2 \label{eq57a}
\end{equation}
with
\begin{equation}
L\left(z\right)\equiv H\left(x\right)=P_{\nu, \lambda}^\mu\left(\cos\theta\right) \label{eq57b}
\end{equation}
leads to the ordinary differential equation
\begin{equation}
4z\left(1-z\right)L''+2\left[1-\left(3+2\lambda\right)z\right]L'+\left[\nu\left(\nu+1\right)-\frac{\mu^2}{1-z}\right]L=0
\label{eq57c}
\end{equation}
\end{subequations}
where the coefficient of the highest order derivative $L''$ in $z$ is cubic instead of quadratic in \eqref{eq42c}. The degree of the coefficient of $L''$ can be depressed further by a change of dependent variable
\begin{subequations}
\begin{equation}
L\left(z\right)=\left(1-z\right)^\vartheta M\left(z\right)
\label{eq58a}
\end{equation}
leading to the ordinary differential equation
\begin{align}
4z\left(1-z\right)M''&+2\left[1-\left(3+2\lambda+4\vartheta\right)z\right]M'\nonumber\\
&+\left[\nu\left(\nu+1\right)-4\vartheta^2-2\vartheta-4\lambda\vartheta-\frac{\mu^2-4\vartheta^2-4\lambda\vartheta}{1-z}\right]M=0
\label{eq58b}
\end{align}
\end{subequations}
where the constant $\vartheta$ can be chosen at will. Choosing $\vartheta$ to satisfy
\begin{subequations}
\begin{equation}
\vartheta^2+\lambda\vartheta-\frac{\mu^2}{4}=0,
\label{eq59a}
\end{equation}
the ordinary differential equation \eqref{eq58b} can be divided through by $1-z$, so that in
\begin{equation}
z\left(1-z\right)M''+\left[\frac{1}{2}-\left(\lambda+\frac{3}{2}+2\vartheta\right)z\right]M'+\left[\frac{\nu\left(\nu+1\right)}{4}-\vartheta^2-\frac{\vartheta}{2}-\lambda\vartheta\right]M=0
\label{eq59b}
\end{equation}
\end{subequations}
the coefficient of the highest order derivative is now of degree two. Furthermore, \eqref{eq59b} is a Gaussian hypergeometric equation, like \eqref{eq45b}, with parameters
\begin{equation}
\gamma=\frac{1}{2}, \quad \alpha+\beta=2\vartheta+\lambda+\frac{1}{2}, \quad \alpha\beta=\vartheta^2+\frac{\vartheta}{2}+\lambda\vartheta-\frac{\nu\left(\nu+1\right)}{4}
\label{eq60a-c}
\end{equation}
with $\vartheta$ given by \eqref{eq59a}. Thus, the function $M$ in \eqref{eq58b} is a Gaussian hypergeometric function, confirming that $z=0, 1, \infty$ are regular singularities. Substitution of \eqref{eq57a}, \eqref{eq57b} and \eqref{eq58a} leads to \eqref{eq55a} or \eqref{eq55b} where the roots of \eqref{eq59a} are in the first equation of \eqref{eq56a-c} and from \eqref{eq60a-c}, with $\alpha, \beta\equiv\chi$, follows
\begin{equation}
0=\chi^2-\left(\alpha+\beta\right)\chi+\alpha\beta=\chi^2-\left(2\vartheta+\lambda+\frac{1}{2}\right)\chi+\vartheta^2+\frac{\vartheta}{2}+\lambda\vartheta-\frac{\nu\left(\nu+1\right)}{4}
\label{eq61b}
\end{equation}
whose roots are in the last two equations of \eqref{eq56a-c}.
\end{proof}

Since there are two roots of \eqref{eq59a}, substitution in the last two equations of \eqref{eq56a-c} leads to two distinct set of parameters ($\vartheta$, $\alpha$, $\beta$); each specifies one solution of the hyperspherical associated Legendre differential equation. Each solution can be written in the alternate form \eqref{eq55a} for \eqref{eq42c} or \eqref{eq55b} for \eqref{eq24b}. Both solutions are plotted in the figure \ref{fig0}. The solid line is the solution using the plus sign in the first equation of \eqref{eq56a-c} while the dashed line is the solution with the minus sign. Since the two particular integrals are linearly independent, the general integral is a linear combination of both.

\section{Conclusion}
\label{sec:6}

The partial differential equation consisting of the Laplacian in $N$ dimensions equated to a linear differential operator in time with constant coefficients was considered in the section \ref{sec:3}, as the extended equation of mathematical physics, that includes frequently used equations such as the wave, heat, Schr\"{o}dinger, telegraph and other equations. The solution, stated in the section \ref{sec:4}, in hyperspherical coordinates by separation of variables leads to known functions except for $N-3$ latitudes. The latitudes are specified by the hyperspherical associated Legendre functions \eqref{eq24a} that reduce to the original associated Legendre functions only in three dimensions. The hyperspherical associated Legendre functions are expressible in terms of Gaussian hypergeometric functions, not only for the original, but also for the hyperspherical Legendre functions, as explained in section \ref{sec:5}.

It has not been necessary to go beyond the Gaussian hypergeometric functions and to consider the extended Gaussian hypergeometric functions \cite{Campos2000a,Campos2001}, which have an irregular singularity at infinity. The brief preceding analysis shows that the hyperspherical associated Legendre functions are a generalization of the associated Legendre functions. Although the usual methods of the theory of special functions and singular differential equations apply \cite{Ince1956,Forsyth1890-1906,Kamke1942}, the properties of the hyperspherical associated Legendre functions can also be obtained from the relation with Gaussian hypergeometric functions.

The hyperspherical associated Legendre functions are a joint generalization of the hyperspherical Legendre functions \cite{Campos2010} that extend the classical multipoles from three to higher dimensions and of the classical associated Legendre functions that lead to spherical harmonics \cite{MacRobert1967,Hobson1931}. A suggested notation for Legendre functions is indicated in the table \ref{tab:1} and figure \ref{fig:1} adding to the classical (i) Legendre and (ii) associated Legendre functions two more: (iii) the hyperspherical associated Legendre functions associated with the separation of the Laplacian in hyperspherical coordinates (present paper); (iv) the hyperspherical Legendre functions arising from the extension of the multipolar expansion to $N$ dimensions \cite{Campos2014}, which are the particular case $\lambda=0$.

\begin{table}[!htb]
  \renewcommand{\arraystretch}{1.2} 
  \centering
\begin{tabular}{cc}
\hline
$P_\nu\left(z\right)$                & Legendre function of variable $z$ and degree $\nu$                                                                                                        \\ \hline
$P_\nu^\mu\left(z\right)$            & Associated Legendre function of variable $z$, order $\mu$ and degree $\nu$                                                                                \\ \hline
$P_{\nu, \lambda}\left(z\right)$     & \begin{tabular}[c]{@{}c@{}}Hvperspherical Legendre function \\ of variable $z$, dimension $\lambda$ and degree $\nu$\end{tabular}              \\ \hline
$P_{\nu, \lambda}^\mu\left(z\right)$ & \begin{tabular}[c]{@{}c@{}}Hyperspherical associated Legendre function of variable $z$, \\ dimension $\lambda$, order $\mu$ and degree $\nu$\end{tabular} \\ \hline
\end{tabular}
  \caption{Notation for Legendre functions.}
  \label{tab:1}
\end{table}

\begin{figure}
\centering
\includegraphics{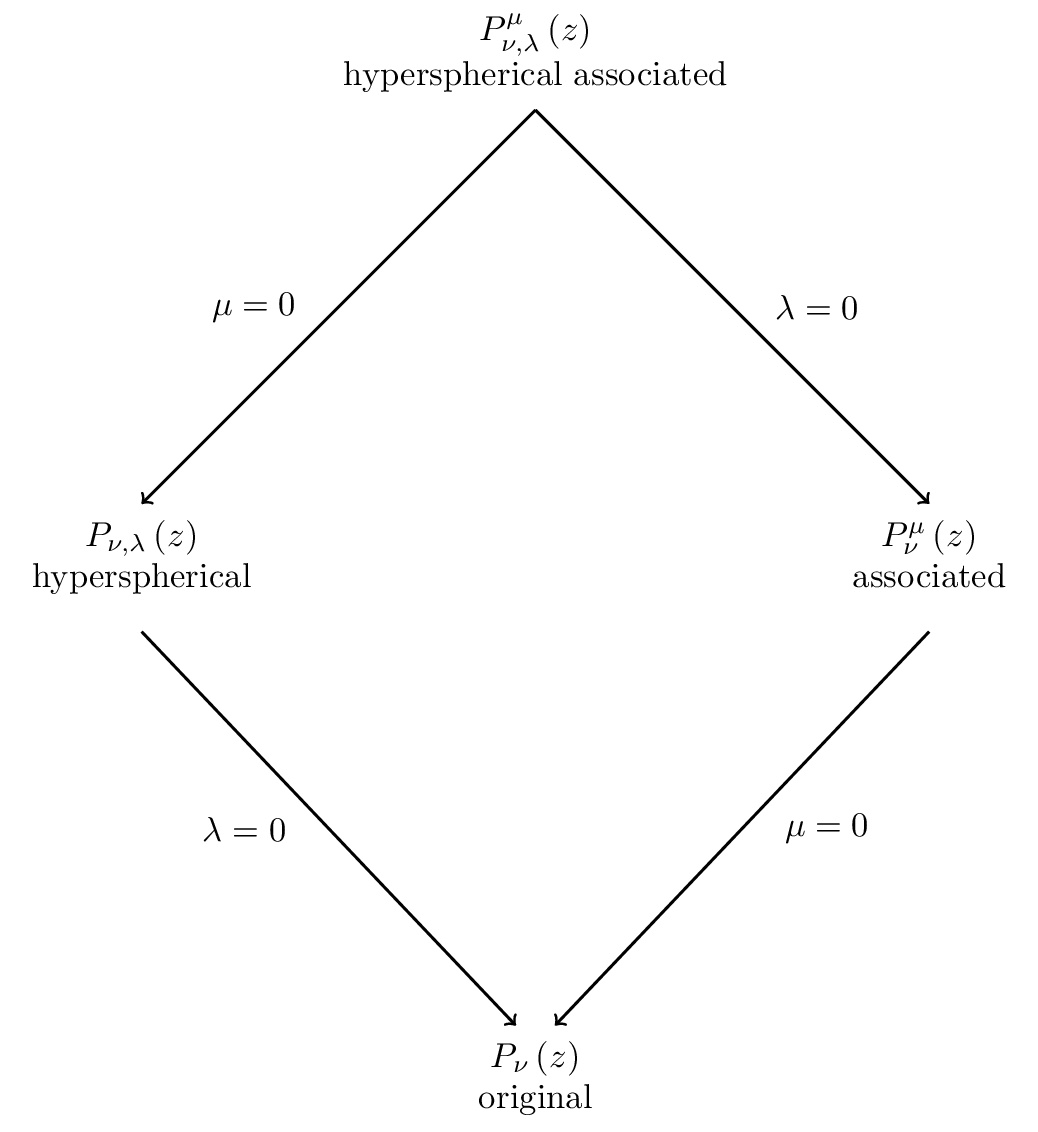}
\caption{Relation between Legendre functions where $P$ denote functions of first kind and $Q$ are functions of second kind.}
\label{fig:1}
\end{figure}

\section*{Acknowledgments}

This work was supported by the Funda\c{c}\~{a}o para a Ci\^{e}ncia e Tecnologia (FCT), Portugal, through Institute of Mechanical Engineering (IDMEC), under the Associated Laboratory for Energy, Transports and Aeronautics (LAETA), whose grant numbers are UID/EMS/50022/2019 and SFRH/BD/143828/2019.

\appendix

\section{Helmholtz equation and its solutions}
\label{sec:appA}

The complete equation \eqref{eq13a} with all terms, as well as its generalization \eqref{eq13b}, can be reduced to an Helmholtz equation by considering the Fourier transform in time
\begin{subequations}
\begin{equation}
F\left(x_n, t\right)=\int\limits_{-\infty}^{+\infty} \widetilde{F}\left(x_n, \omega\right) \mathrm{e}^{-\mathrm{i}\omega t}\,\mathrm{d}\omega
\label{eq8.342a}
\end{equation}
implying that the derivation with regard to time
\begin{equation}
\frac{\partial F}{\partial t}\rightarrow -\mathrm{i}\omega\widetilde{F}
\label{eq8.342b}
\end{equation}
\end{subequations}
is equivalent to multiplication by $-\mathrm{i}\omega$ where $\omega$ is the frequency. Substitution of \eqref{eq8.342b} in \eqref{eq13b} shows that the solution of the generalized isotropic equation of mathematical physics \eqref{eq13b} has a Fourier transform in time \eqref{eq8.342a} that satisfies an Helmholtz equation
\begin{subequations}
\begin{equation}
\nabla^2\widetilde{F}+k^2\widetilde{F}=0
\label{eq8.343b}
\end{equation}
with the square of wavenumber
\begin{equation}
k^2=\sum_{l=0}^L A_l\left(-\mathrm{i}\omega\right)^l
\label{eq8.343a}
\end{equation}
\end{subequations}
that is complex in the presence of time derivatives of odd order, for example
\begin{equation}
k^2=-a+\frac{\mathrm{i}\omega}{b}+\frac{\omega^2}{c^2}
\label{eq8.344a}
\end{equation}
for the spectrum of the second-order isotropic equation of mathematical physics \eqref{eq13a}.

Before determining the solutions of the Helmholtz equation \eqref{eq8.343b} in cylindrical, spherical or hypercylindrical coordinates that is the new feature of this work, the appendix \ref{sec:appA.1} intends to determine the solutions in Cartesian coordinates to show that the only solutions in that case are the sinusoidal functions. The solution of the Helmholtz equation is therefore simplest in Cartesian coordinates for $N$ dimensions when the Laplacian is a second-order partial differential operator with constant coefficients. When the curvilinear coordinates are used in the Laplacian, special functions appear in the solution.

\subsection{Separation of variables in Cartesian coordinates}
\label{sec:appA.1}

Using the Laplacian in $N$-dimensional Cartesian coordinates,
\begin{subequations}
\begin{equation}
\nabla^2\widetilde{F}=\sum_{n=1}^N \frac{\partial^2\widetilde{F}}{\partial x_n^2},
\label{eq8.345a}
\end{equation}
the Helmholtz equation \eqref{eq8.343b} becomes
\begin{equation}
\sum_{n=1}^N \frac{\partial^2\widetilde{F}}{\partial x_n^2}=-k^2\widetilde{F}
\label{eq8.345b}
\end{equation}
whose solution may be sought by the method of separation of variables as the product
\begin{equation}
\widetilde{F}\left(x_1, \ldots, x_N, \omega\right)=\prod_{n=1}^N X_n\left(x_n\right)
\label{eq8.345c}
\end{equation}
\end{subequations}
of $N$ functions, one of each variable. Substituting \eqref{eq8.345c} in \eqref{eq8.345b} and dividing by $\widetilde{F}$ leads to
\begin{subequations}
\begin{equation}
\sum_{n=1}^N \frac{1}{X_n}\frac{\mathrm{d}^2 X_n}{\mathrm{d} x_n^2}=-k^2.
\label{eq8.346a}
\end{equation}
Since each term on the l.h.s. of \eqref{eq8.346a} depends on a different variable they must all be constant,
\begin{equation}
\frac{1}{X_n}\frac{\mathrm{d}^2 X_n}{\mathrm{d} x_n^2}=-k_n^2, \label{eq8.346b}
\end{equation}
with their sum satisfying
\begin{equation}
k^2=\sum_{n=1}^N k_n^2. \label{eq8.346c}
\end{equation}
\end{subequations}
Thus, to the position vector of coordinates $x_n$ may be associated a wavevector with components $k_n$ whose the sum of squares specifies the wavenumber $k$. The differential equation \eqref{eq8.346b} has the solutions
\begin{equation}
X_n\left(x_n\right)=B_n^+\exp\left(\mathrm{i} k_n x_n\right)+B_n^-\exp\left(-\mathrm{i} k_n x_n\right), \label{eq8.347b}
\end{equation}
that, substituted in \eqref{eq8.345c}, lead to
\begin{align}
\widetilde{F}\left(x_1, \ldots, x_N, \omega\right)&=\prod_{n=1}^N C_n\cos\left(k_n x_n-\alpha_n\right)\nonumber\\
&=\frac{1}{2^N}\prod_{n=1}^N C_n\left\lbrace\exp\left[\mathrm{i}\left(k_n x_n-\alpha_n\right)\right]+\exp\left[-\mathrm{i}\left(k_n x_n-\alpha_n\right)\right]\right\rbrace \label{eq8.347d}
\end{align}
where $B_n^\pm$ in \eqref{eq8.347b} and $\left(C_n, \alpha_n\right)$ in \eqref{eq8.347d} are alternative pairs of arbitrary constants related by
\begin{subequations}
\begin{align}
2B_n^\pm&=C_n\exp\left(\mp\mathrm{i}\alpha_n\right): \label{eq8.348a-b} \\
\exp\left(2\mathrm{i}\alpha_n\right)&=\frac{B_n^-}{B_n^+}, \label{eq8.348c} \\
C_n&=B_n^+\exp\left(\mathrm{i}\alpha_n\right)+B_n^-\exp\left(-\mathrm{i}\alpha_n\right). \label{eq8.348d}
\end{align}
\end{subequations}
Thus, the Helmholtz equation \eqref{eq8.343b} in $N$-dimensional Cartesian coordinates \eqref{eq8.345b} has the solution in terms of sinusoidal functions alone \eqref{eq8.347d} where $B_n^\pm$ and $\left(C_n, \alpha_n\right)$ are alternative pairs of arbitrary constants related by \eqref{eq8.348a-b} to \eqref{eq8.348d}.

In the case of polar or spherical coordinates in the plane or space respectively, the Laplacian has variable coefficients and the method of separation of variables leads to linear ordinary differential equations with variable coefficients whose the solutions involve special functions. The solution of the Helmholtz equation
is reviewed briefly for cylindrical coordinates in the appendix \ref{sec:appA.2} and for spherical coordinates in the appendix \ref{sec:appA.3} as a precursor to the solution for hypercylindrical coordinates in the appendix \ref{sec:appB}. The solution of the Helmholtz equation in cylindrical coordinates requires Bessel functions, in spherical coordinates the associated Legendre functions are needed and in hypercylindrical coordinates the hyperspherical associated Legendre functions appear.

\subsection{Cylindrical coordinates and Bessel functions}
\label{sec:appA.2}

The Laplacian using cylindrical coordinates in space adds to the Laplacian in polar coordinates, that span a plane, an orthogonal Cartesian coordinate $z$, leading to the Helmholtz equation
\begin{equation}
\frac{1}{r}\left(r\frac{\partial\widetilde{F}}{\partial r}\right)+\frac{1}{r^2}\frac{\partial^2\widetilde{F}}{\partial\phi^2}+\frac{\partial^2\widetilde{F}}{\partial z^2}=-k^2\widetilde{F}. \label{eq8.349}
\end{equation}
In this appendix \ref{sec:appA.2}, $r$ denotes the distance to the polar axis. The last Helmholtz equation can be solved by separation of variables
\begin{subequations}
\begin{equation}
\widetilde{F}=R\left(r\right)\,\Phi\left(\phi\right)\,Z\left(z\right) \label{eq8.350a}
\end{equation}
leading to
\begin{equation}
\frac{1}{Rr}\frac{\mathrm{d}}{\mathrm{d} r}\left(r\frac{\mathrm{d} R}{\mathrm{d} r}\right)+\frac{1}{r^2\Phi}\frac{\mathrm{d}^2\Phi}{\mathrm{d}\phi^2}+\frac{1}{Z}\frac{\mathrm{d}^2 Z}{\mathrm{d} z^2}+k^2=0. \label{eq8.350b}
\end{equation}
\end{subequations}
The first three terms of \eqref{eq8.350b} depend on different variables leading to the constants
\begin{subequations}
\begin{align}
\frac{1}{\Phi}\frac{\mathrm{d}^2\Phi}{\mathrm{d}\phi^2}&=-m^2, \label{eq8.351a} \\
\frac{1}{Z}\frac{\mathrm{d}^2Z}{\mathrm{d} z^2}&=-K^2, \label{eq8.351b}
\end{align}
that must satisfy
\begin{equation}
\frac{1}{Rr}\frac{\mathrm{d}}{\mathrm{d} r}\left(r\frac{\mathrm{d} R}{\mathrm{d} r}\right)-\frac{m^2}{r^2}+k^2-K^2=0. \label{eq8.351c}
\end{equation}
\end{subequations}
The solutions of \eqref{eq8.351a} and \eqref{eq8.351b} are sinusoidal functions, respectively
\begin{subequations}
\begin{align}
\Phi\left(\phi\right)&=C_+\mathrm{e}^{\mathrm{i} m\phi}+C_-\mathrm{e}^{-\mathrm{i} m\phi} \label{eq8.352a}\\
Z\left(z\right)&=B_+\mathrm{e}^{\mathrm{i} Kz}+B_-\mathrm{e}^{-\mathrm{i} Kz}, \label{eq8.352b}
\end{align}
\end{subequations}
where $C_\pm$ and $B_\pm$ are arbitrary constants. The differential equation \eqref{eq8.351c}, equivalent to
\begin{equation}
r^2\frac{\mathrm{d}^2R}{\mathrm{d} r^2}+r\frac{\mathrm{d} R}{\mathrm{d} r}+\left(\overline{k}^2r^2-m^2\right)R=0, \label{eq8.353b}
\end{equation}
is a cylindrical Bessel differential equation whose solution is a linear combination of Bessel $J_m$ and Neumann $Y_m$ functions,
\begin{subequations}
\begin{equation}
R\left(r\right)=E_+J_m\left(\overline{k}r\right)+E_-Y_m\left(\overline{k}r\right), \label{eq8.354b}
\end{equation}
with: (i) integer order $m$ equal to azimuthal wavenumber $m$ in \eqref{eq8.352a}; (ii) variable $\overline{k}r$ involving the radial wavenumber $\overline{k}$ which is equal to
\begin{equation}
\overline{k}=\left|
k^2-K^2
\right|^{1/2}\Leftrightarrow k^2=K^2+\overline{k}^2; \label{eq8.354a}
\end{equation}
\end{subequations}
(iii) the sum of the squares of the radial wavenumber $\overline{k}$ in \eqref{eq8.354b} and of the axial wavenumber $K$ in \eqref{eq8.352b} being the square of the total wavenumber $k$ in the Helmholtz equation \eqref{eq8.349}.

Substituting the solutions of $\Phi$, $Z$ and $X$ in \eqref{eq8.350a}, it follows that the solution of the Helmholtz equation in cylindrical coordinates \eqref{eq8.349} is
\begin{align}
\widetilde{F}\left(r, \phi, z, \omega\right)&=\left(B_+\mathrm{e}^{\mathrm{i} Kz}+B_-\mathrm{e}^{-\mathrm{i} Kz}\right)\left(C_+\mathrm{e}^{\mathrm{i} m\phi}+C_-\mathrm{e}^{-\mathrm{i} m\phi}\right)\nonumber\\
&\times\left[E_+J_m\left(\overline{k}r\right)+E_-Y_m\left(\overline{k}r\right)\right] \label{eq8.355}
\end{align}
involving: (i) three pairs of arbitrary constants $\left(B_\pm, C_\pm, E_\pm\right)$; (ii) the products of sinusoidal functions of the axial $z$ and azimuthal $\phi$ coordinates by a linear combination of Bessel $J_m$ and Neumann $Y_m$ functions; (iii) the azimuthal $m$, axial $K$, radial $\overline{k}$ and total $k$ wavenumbers with the last three related by \eqref{eq8.354b}. In the case of polar coordinates in the plane, the dependence on the axial coordinate $z$ in the Laplacian operator is omitted, so that the Helmholtz equation \eqref{eq8.349} simplifies to
\begin{subequations}
\begin{equation}
\frac{1}{r}\frac{\partial}{\partial r}\left(r\frac{\partial\widetilde{F}}{\partial r}\right)+\frac{1}{r^2}\frac{\partial^2\widetilde{F}}{\partial\phi^2}+k^2\widetilde{F}=0. \label{eq8.356a}
\end{equation}
The corresponding axial wavenumber is zero, $K=0$, omitting one of the factors in the solution \eqref{eq8.355} when passing to
\begin{equation}
\widetilde{F}\left(r, \phi, \omega\right)=\left(C_+\mathrm{e}^{\mathrm{i} m\phi}+C_-\mathrm{e}^{-\mathrm{i} m\phi}\right)\left[E_+J_m\left(kr\right)+E_-Y_m\left(kr\right)\right] \label{eq8.356b}
\end{equation}
\end{subequations}
where the total and radial wavenumbers coincide. The solution of the Helmholtz equation in spherical coordinates involves besides sinusoidal and Bessel functions also the associated Legendre functions.

\subsection{Spherical coordinates and associated Legendre functions}
\label{sec:appA.3}

The Laplace operator in spherical coordinates leads to the Helmholtz equation
\begin{equation}
\frac{1}{r^2}\frac{\partial}{\partial r}\left(r^2\frac{\partial\widetilde{F}}{\partial r}\right)+\frac{1}{r^2\sin\theta}\frac{\partial}{\partial\theta}\left(\sin\theta\frac{\partial\widetilde{F}}{\partial\theta}\right)+\frac{1}{r^2\sin^2\theta}\frac{\partial^2\widetilde{F}}{\partial\phi^2}+k^2\widetilde{F}=0. \label{eq8.357}
\end{equation}
The solution by separation of variables
\begin{subequations}
\begin{equation}
\widetilde{F}\left(r, \theta, \phi, \omega\right)=R\left(r\right)\,\Theta\left(\theta\right)\,\Phi\left(\phi\right) \label{eq8.358a}
\end{equation}
leads to
\begin{equation}
\frac{1}{R}\frac{\mathrm{d}}{\mathrm{d} r}\left(r^2\frac{\mathrm{d} R}{\mathrm{d} r}\right)+\frac{1}{\Theta\sin\theta}\frac{\mathrm{d}}{\mathrm{d}\theta}\left(\sin\theta\frac{\mathrm{d}\Theta}{\mathrm{d}\theta}\right)+\frac{1}{\Phi\sin^2\theta}\frac{\mathrm{d}^2\Phi}{\mathrm{d}\phi^2}+k^2r^2=0 \label{eq3.358b}
\end{equation}
\end{subequations}
that is satisfied by three separate ordinary differential equations specifying the: (i) azimuthal dependence \eqref{eq8.351a} in terms of sinusoidal functions \eqref{eq8.352a}; (ii) latitudinal dependence
\begin{subequations}
\begin{equation}
\frac{1}{\Theta\sin\theta}\frac{\mathrm{d}}{\mathrm{d}\theta}\left(\sin\theta\frac{\mathrm{d}\Theta}{\mathrm{d}\theta}\right)-\frac{m^2}{\sin^2\theta}=-q\left(q+1\right) \label{eq8.359a}
\end{equation}
leading to an associated Legendre differential equation
\begin{equation}
\frac{\mathrm{d}^2\Theta}{\mathrm{d}\theta^2}+\cot\theta\frac{\mathrm{d}\Theta}{\mathrm{d}\theta}+\left[q\left(q+1\right)-m^2\csc^2\theta\right]\Theta=0 \label{eq8.359b}
\end{equation}
\end{subequations}
whose the solution
\begin{equation}
\Theta\left(\theta\right)=D_+P_q^m\left(\cos\theta\right)+D_-Q_q^m\left(\cos\theta\right) \label{eq8.360}
\end{equation}
is a linear combination of the first $P_q^m$ and second $Q_q^m$ kinds of associated Legendre functions with degree $q$ and order $m$; (iii) radial dependence
\begin{align}
\frac{1}{R}\frac{\mathrm{d}}{\mathrm{d} r}\left(r^2\frac{\mathrm{d} R}{\mathrm{d} r}\right)+k^2r^2-q\left(q+1\right)&=r^2\frac{\mathrm{d}^2 R}{\mathrm{d} r^2}+2r\frac{\mathrm{d} R}{\mathrm{d} r}+\left[k^2r^2-q\left(q+1\right)\right]\nonumber\\
&=0 \label{eq8.361b}
\end{align}
specified by a spherical Bessel differential equation whose the general integral
\begin{equation}
R\left(r\right)=E_+j_q\left(kr\right)+E_-y_q\left(kr\right) \label{eq8.362}
\end{equation}
is a linear combination of spherical Bessel $j_q$ and Neumann $y_q$ functions of order $q$ where $k$ is the radial wavenumber. Substituting the solutions of $\Phi$, $\Theta$ and $R$ in \eqref{eq8.358a}, it follows that the solution of the Helmholtz equation in spherical coordinates \eqref{eq8.357} is
\begin{align}
\widetilde{F}\left(r, \theta, \phi, \omega\right)&=\left(C_+\mathrm{e}^{\mathrm{i} m\phi}+C_-\mathrm{e}^{-\mathrm{i} m\phi}\right)\left[D_+P_q^m\left(\cos\theta\right)+D_-Q_q^m\left(\cos\theta\right)\right]\nonumber\\
&\times\left[E_+j_q\left(kr\right)+E_-y_q\left(kr\right)\right] \label{eq8.363}
\end{align}
involving: (i) three pairs of arbitrary constants of integration $\left(C_\pm, D_\pm, E_\pm\right)$; (ii) sinusoidal functions of longitude \eqref{eq8.352a} with wavenumber $m$; (iii) two kinds of associated Legendre functions \eqref{eq8.360} of the cosine of the latitude with order $m$ and degree $q$; (iv) spherical Bessel and Neumann functions \eqref{eq8.362} of order $q$ and variable $kr$ where the radial distance is multiplied by the radial wavenumber. In this case, there is only one latitude and the corresponding associated Legendre equation. In hyperspherical coordinates, there are more latitudes and therefore more (hyperspherical, except one) associated Legendre equations, each one of degree $q_1, q_2,\ldots, q_{N-2}$ and, in that case, the constants introduced lead to $q_1=q$ where $q$ is associated to the order of spherical Bessel functions (explained in the subsection \ref{sec:3.2}). In the case of spherical coordinates, with $N=3$, there is only one associated Legendre equation of degree $q_1$ that is also simultaneously equal to the order of spherical Bessel functions $q$. Consequently, the degree of the associated Legendre functions and the order of spherical Bessel and Neumann functions, in the case of spherical coordinates, are the same because there is only one latitude.

The solution of the Helmholtz equation in cylindrical and spherical coordinates for three dimensions can be generalized to hypercylindrical and hyperspherical coordinates respectively for any higher number of dimensions. The generalization to hypercylindrical coordinates and the solution of the Helmholtz equation in that system of coordinates are detailed in the appendix \ref{sec:appB}.

\section{Multidimensional hypercylindrical coordinates}
\label{sec:appB}

The solution of the Helmholtz equation in hypercylindrical coordinates involves six steps: (i) the relation with $N$-dimensional Cartesian coordinates; (ii) the orthogonal base vectors with their moduli specifying the scale factors; (iii) the Laplacian operator leading to the Helmholtz equation; (iv) the separation of variables leading to sinusoidal functions in the azimuthal and axial directions, and Bessel functions in the radial direction; (v) the associated Legendre functions that appear for the first latitude and the hyperspherical associated Legendre functions that appear for higher-order latitudes; (vi) in the case of hypercylindrical coordinates, the last latitude that is replaced by a Cartesian axial coordinate leading to sinusoidal functions.

The hypercylindrical coordinates, with the radius $0\leq r<\infty$, longitude $0\leq\phi\leq 2\pi$, axial distance $-\infty<z<+\infty$ and $N-3$ latitudes $0\leq\theta_1, \ldots, \theta_{N-3}\leq\pi$, are defined by $N$ relations with the $N$-dimensional Cartesian coordinates. Most of them are similar with respect to hyperspherical coordinates. For instance, $N-3$ relations are similar in terms of the distance from the axis $r$ (in hyperspherical coordinates, $r$ denotes the distance to the origin) and $N-3$ latitudes, $\theta_1, \ldots, \theta_{N-3}$, leading to the same relations as \eqref{eq2}, except for the last three equations. Besides the $N-3$ relations, the hyperspherical coordinates have one more latitude $\theta_{N-2}$ and a longitude $\phi$, leading to the last three equations of \eqref{eq2}. Otherwise, the hypercylindrical coordinates have one axial Cartesian coordinate $z$ and a longitude $\phi$, substituting the last three equations of \eqref{eq2} by
\begin{equation}
\begin{aligned}
x_{N-2}&=r\sin\theta_1\sin\theta_2\ldots\sin\theta_{N-3}\cos\phi,\\
x_{N-1}&=r\sin\theta_1\sin\theta_2\ldots\sin\theta_{N-3}\sin\phi,\\
x_N&=z.
\label{eq8.366f'-h'}
\end{aligned}
\end{equation}
Thus, the transformation from hypercylindrical to Cartesian coordinates in $N$-dimensions is given by \eqref{eq2}, except the last three equations, plus \eqref{eq8.366f'-h'}. In two dimensions, $N=2$, with $x_1\equiv x$ and $x_2\equiv y$, this leads to polar coordinates ($r$, $\phi$). In three dimensions, $N=3$, with $x_3\equiv z$, this leads to cylindrical coordinates ($r$, $\phi$, $z$). The name hypercylindrical coordinates arises because the first coordinate hypersurface is an hypercylinder.

These transformations can be inverted from $N$-dimensional Cartesian to hypercylindrical coordinates: the distance from the axis is
\begin{subequations}
\begin{equation}
r=\left|\left(x_1\right)^2+\left(x_2\right)^2+\ldots+\left(x_{N-1}\right)^2\right|^{1/2}
\label{eq8.367a'}
\end{equation}
showing that the coordinate hypersurface $r=\mathrm{const}$ is an hypercylinder of radius $r$; the next $N-3$ relations are similar to hyperspherical coordinates, with
\begin{equation}
\begin{aligned}
&\cot\theta_1=x_1\left\lbrace\left|
\left(x_2\right)^2+\ldots+\left(x_{N-1}\right)^2
\right|^{-1/2}
\right\rbrace,\\
&\cot\theta_2=x_2\left\lbrace\left|
\left(x_3\right)^2+\ldots+\left(x_{N-1}\right)^2
\right|^{-1/2}\right\rbrace,\\
&\vdots\\
&\cot\theta_n=x_n\left\lbrace\left|
\left(x_{n+1}\right)^2+\ldots+\left(x_{N-1}\right)^2
\right|^{-1/2}\right\rbrace,\\
&\vdots\\
&\cot\theta_{N-3}=x_{N-3}\left\lbrace\left|
\left(x_{N-2}\right)^2+\left(x_{N-1}\right)^2
\right|^{-1/2}\right\rbrace;
\end{aligned}
\label{eq8.367b'-e'}
\end{equation}
the last two relations are different because
\begin{equation}
\begin{aligned}
&\cot\phi=\frac{x_{N-2}}{x_{N-1}},\\
&z=x_N.
\end{aligned}
\label{eq8.367f'-g'}
\end{equation}
\end{subequations}
In the transformations \eqref{eq8.367b'-e'} and \eqref{eq8.367f'-g'}, it would be possible to substitute the cotangent by other circular functions. The transformations \eqref{eq8.367a'} to \eqref{eq8.367b'-e'} from $N$-dimensional Cartesian coordinates to hypercylindrical coordinates are the inverses of \eqref{eq2} (not regarding the last three equations) and \eqref{eq8.366f'-h'}, and show that the hypercylindrical coordinates are orthogonal on the hypercylinder $r=\mathrm{const}$. The orthogonality of the hypercylindrical coordinates is proved from the base vectors that leads to scale factors.

\subsection{Base vectors and scale factors}

The Cartesian components of the contravariant hypercylindrical base vectors follow from the transformation from hypercylindrical to Cartesian coordinates,
\begin{equation}
\begin{aligned}
\overrightarrow{e}_r&=\frac{\partial x_i}{\partial r}=\left\lbrace\cos\theta_1, \sin\theta_1\cos\theta_2, \sin\theta_1\sin\theta_2\cos\theta_3, \ldots,\right.\\
&\sin\theta_1\sin\theta_2\sin\theta_3\ldots\sin\theta_{n-1}\cos\theta_n, \ldots, \sin\theta_1\sin\theta_2\sin\theta_3\ldots\sin\theta_{N-4}\cos\theta_{N-3},\\
&\left.  \sin\theta_1\sin\theta_2\sin\theta_3\ldots\sin\theta_{N-3}\cos\phi, \sin\theta_1\sin\theta_2\sin\theta_3\ldots\sin\theta_{N-3}\sin\phi, 0\right\rbrace,\\
\overrightarrow{e}_1&=\frac{\partial x_i}{\partial\theta_1}=r\left\lbrace -\sin\theta_1, \cos\theta_1\cos\theta_2, \cos\theta_1\sin\theta_2\cos\theta_3, \ldots,\right.\\
& \cos\theta_1\sin\theta_2\sin\theta_3\ldots\sin\theta_{n-1}\cos\theta_n,\ldots, \cos\theta_1\sin\theta_2\sin\theta_3\ldots\sin\theta_{N-4}\cos\theta_{N-3},\\
&\left.  \cos\theta_1\sin\theta_2\sin\theta_3\ldots\sin\theta_{N-3}\cos\phi, \cos\theta_1\sin\theta_2\sin\theta_3\ldots\sin\theta_{N-3}\sin\phi, 0\right\rbrace,\\
\overrightarrow{e}_2&=\frac{\partial x_i}{\partial\theta_2}=r\sin\theta_1\left\lbrace 0, -\sin\theta_2, \cos\theta_2\cos\theta_3, \cos\theta_2\sin\theta_3\cos\theta_4, \ldots,\right.\\
&\cos\theta_2\sin\theta_3\sin\theta_4\ldots\sin\theta_{n-1}\cos\theta_n,\ldots, \cos\theta_2\sin\theta_3\sin\theta_4\ldots\sin\theta_{N-4}\cos\theta_{N-3},\\
&\left.  \cos\theta_2\sin\theta_3\sin\theta_4\ldots\sin\theta_{N-3}\cos\phi, \cos\theta_2\sin\theta_3\sin\theta_4\ldots\sin\theta_{N-3}\sin\phi, 0\right\rbrace,\\
&\vdots\\
\overrightarrow{e}_n&=\frac{\partial x_i}{\partial\theta_n}=r\sin\theta_1\sin\theta_2\ldots\sin\theta_{n-1}\left\lbrace 0, 0,\ldots, 0,-\sin\theta_n, \cos\theta_n\cos\theta_{n+1},\right.\\
& \cos\theta_n\sin\theta_{n+1}\cos\theta_{n+2},\ldots, \cos\theta_n\sin\theta_{n+1}\ldots\sin\theta_{N-4}\cos\theta_{N-3},\\
&\left. \cos\theta_n\sin\theta_{n+1}\ldots\sin\theta_{N-3}\cos\phi, \cos\theta_n\sin\theta_{n+1}\ldots\sin\theta_{N-3}\sin\phi, 0\right\rbrace,\\
&\vdots\\
\overrightarrow{e}_{N-3}&=\frac{\partial x_i}{\partial\theta_{N-3}}=r\sin\theta_1\sin\theta_2\ldots\sin\theta_{N-4}\left\lbrace 0, 0,\ldots, 0, -\sin\theta_{N-3}, \vphantom{\cos\theta_{N-3}}\right.\\
&\left.\cos\theta_{N-3}\cos\phi, \cos\theta_{N-3}\sin\phi, 0\right\rbrace,\\
\overrightarrow{e}_\phi&=\frac{\partial x_i}{\partial\phi}=r\sin\theta_1\sin\theta_2\ldots\sin\theta_{N-3}\left\lbrace 0,0,\ldots,0, -\sin\phi, \cos\phi, 0\right\rbrace,\\
\overrightarrow{e}_{z}&=\frac{\partial x_i}{\partial z}=\left\lbrace 0, 0, \ldots, 0, 1\right\rbrace,
\end{aligned}
\label{eq8.368a'-f'}
\end{equation}
showing that all base vectors are orthogonal because
\begin{align}
\forall\, n=1,\ldots, N-3, \quad \overrightarrow{e}_r\cdot\overrightarrow{e}_n&=\overrightarrow{e}_\phi\cdot\overrightarrow{e}_n=\overrightarrow{e}_r\cdot\overrightarrow{e}_\phi=\overrightarrow{e}_r\cdot\overrightarrow{e}_z\nonumber\\
&=\overrightarrow{e}_n\cdot\overrightarrow{e}_z=\overrightarrow{e}_\phi\cdot\overrightarrow{e}_z=0.
\label{eq8.369c}
\end{align}

The hypercylindrical coordinates are an orthogonal curvilinear coordinate system in $N$ dimensions, and the modulus of the base vectors specify the scale factors,
\begin{align}
\forall\, i=1,\ldots,N \quad h_i\equiv\left|
\overrightarrow{e}_i\right|=\left\lbrace 1, r, r\sin\theta_1, r\sin\theta_1\sin\theta_2,\ldots,\vphantom{\sin_{N-2}}\right. \nonumber\\
\left.r\sin\theta_1\sin\theta_2\ldots\sin\theta_{i-2},\ldots,r\sin\theta_1\sin\theta_2\ldots\sin\theta_{N-3}, 1\right\rbrace.
\label{eq8.370a-b}
\end{align}
Comparing to hyperspherical coordinates, the scale factors in hypercylindrical coordinates are similar for the first $N-1$ elements with the distance from the origin replaced by the distance from the axis while the last element is unity because it corresponds to a Cartesian coordinate. The scale factors for an orthogonal curvilinear coordinate system specify through their squares the diagonal of the covariant metric tensor and its determinant \cite{Sokolnikoff1951}. The determinant of the covariant metric tensor is given for hypercylindrical coordinates by
\begin{align}
\left|g\right|^{1/2}&=\prod_{i=1}^N h_i\nonumber\\
&=r^{N-2}\sin^{N-3}\theta_1\sin^{N-4}\theta_2\ldots\sin^{N-i}\theta_{i-2}\ldots\sin^2\theta_{N-4}\sin\theta_{N-3}.
\label{eq8.372'}
\end{align}
The scale factors of an orthogonal curvilinear coordinate system in $N$ dimensions specify the Laplacian operator and hence the Helmholtz equation in hypercylindrical coordinates.

\subsection{Helmholtz equation in hypercylindrical coordinates}

The contravariant metric tensor and the determinant of the covariant metric tensor can be used to write any invariant differential operator, for example the scalar Laplacian, that simplifies to \eqref{eq11b} in orthogonal curvilinear coordinates \cite{Sokolnikoff1951}. Using the scale factors \eqref{eq8.370a-b} in hypercylindrical coordinates, the successive terms are: (i) for the radius,
\begin{subequations}
\begin{equation}
\frac{1}{\sqrt{g}}\frac{\partial}{\partial r}\left[\sqrt{g}\left(h_1\right)^{-2}\frac{\partial}{\partial r}\right]=\frac{1}{r^{N-2}}\frac{\partial}{\partial r}\left(r^{N-2}\frac{\partial}{\partial r}\right)=\frac{\partial^2}{\partial r^2}+\frac{N-2}{r}\frac{\partial}{\partial r},
\label{eq8.374a'}
\end{equation}
which coincides with the radial part of the Laplacian in polar coordinates for $N=2$ or cylindrical coordinates for $N=3$; (ii) for the first latitude,
\begin{equation}
\frac{1}{\sqrt{g}}\frac{\partial}{\partial \theta_1}\left[\sqrt{g}\left(h_2\right)^{-2}\frac{\partial}{\partial \theta_1}\right]=\frac{1}{r^2\sin^{N-3}\theta_1}\frac{\partial}{\partial\theta_1}\left(\sin^{N-3}\theta_1\frac{\partial}{\partial\theta_1}\right);
\label{eq8.374b'}
\end{equation}
(iii) for the second latitude,
\begin{equation}
\frac{1}{\sqrt{g}}\frac{\partial}{\partial \theta_2}\left[\sqrt{g}\left(h_3\right)^{-2}\frac{\partial}{\partial \theta_2}\right]=\frac{1}{r^2\sin^2\theta_1\sin^{N-4}\theta_2}\frac{\partial}{\partial\theta_2}\left(\sin^{N-4}\theta_2\frac{\partial}{\partial\theta_2}\right);
\label{eq8.374c'}
\end{equation}
(iv) for the $n$-th latitude,
\begin{align}
\forall\,n=1,\ldots,N-2: \quad \frac{1}{\sqrt{g}}\frac{\partial}{\partial \theta_n}\left[\sqrt{g}\left(h_{n+1}\right)^{-2}\frac{\partial}{\partial \theta_n}\right]\nonumber\\
=\frac{1}{r^2\sin^2\theta_1\ldots\sin^2\theta_{n-1}\sin^{N-n-2}\theta_n}\frac{\partial}{\partial\theta_n}\left(\sin^{N-n-2}\theta_n\frac{\partial}{\partial\theta_n}\right);
\label{eq8.374d'}
\end{align}
(v) for the last or ($N-3$)-th latitude,
\begin{align}
\frac{1}{\sqrt{g}}\frac{\partial}{\partial \theta_{N-3}}\left[\sqrt{g}\left(h_{N-2}\right)^{-2}\frac{\partial}{\partial \theta_{N-3}}\right]&=\frac{1}{r^2\sin^2\theta_1\ldots\sin^2\theta_{N-4}\sin\theta_{N-3}}\nonumber\\
&\times\frac{\partial}{\partial\theta_{N-3}}\left(\sin\theta_{N-3}\frac{\partial}{\partial\theta_{N-3}}\right);
\label{eq8.374e'}
\end{align}
(vi) for the longitude,
\begin{equation}
\frac{1}{\sqrt{g}}\frac{\partial}{\partial\phi}\left[\sqrt{g}\left(h_{N-1}\right)^{-2}\frac{\partial}{\partial\phi}\right]=\frac{1}{r^2\sin^2\theta_1\ldots\sin^2\theta_{N-3}}\frac{\partial^2}{\partial\phi^2},
\label{eq8.374f'}
\end{equation}
which is again familiar for polar ($N=2$) and cylindrical ($N=3$) coordinates;
(vii) for the axial coordinate,
\begin{equation}
\frac{1}{\sqrt{g}}\frac{\partial}{\partial z}\left[\sqrt{g}\left(h_{N}\right)^{-2}\frac{\partial}{\partial z}\right]=\frac{\partial^2}{\partial z^2},
\label{eq8.374g'}
\end{equation}
\end{subequations}
which is the same for cylindrical ($N=3$) coordinates.

The comparison of the successive forms in hypercylindrical coordinates to cylindrical or polar coordinates can be made observing the l.h.s. of \eqref{eq8.349}. Substituting \eqref{eq8.374a'} to \eqref{eq8.374g'} in the Laplacian \eqref{eq11b} and then in \eqref{eq8.343b} specifies the Helmholtz equation in hypercylindrical coordinates
\begin{align}
-k^2\widetilde{F}=\nabla^2\widetilde{F}&=\frac{1}{r^{N-2}}\frac{\partial}{\partial r}\left(r^{N-2}\frac{\partial\widetilde{F}}{\partial r}\right)+\frac{1}{r^2\sin^{N-3}\theta_1}\frac{\partial}{\partial\theta_1}\left(\sin^{N-3}\theta_1\frac{\partial\widetilde{F}}{\partial\theta_1}\right)\nonumber\\
&+\frac{1}{r^2\sin^{2}\theta_1\sin^{N-4}\theta_2}\frac{\partial}{\partial\theta_2}\left(\sin^{N-4}\theta_2\frac{\partial\widetilde{F}}{\partial\theta_2}\right)+\ldots\nonumber\\
&+\frac{1}{r^2\sin^{2}\theta_1\ldots\sin^{2}\theta_{n-1}\sin^{N-n-2}\theta_n}\frac{\partial}{\partial\theta_n}\left(\sin^{N-n-2}\theta_n\frac{\partial\widetilde{F}}{\partial\theta_n}\right)+\ldots\nonumber\\
&+\frac{1}{r^2\sin^{2}\theta_1\ldots\sin^2\theta_{N-4}\sin\theta_{N-3}}\frac{\partial}{\partial\theta_{N-3}}\left(\sin\theta_{N-3}\frac{\partial\widetilde{F}}{\partial\theta_{N-3}}\right)\nonumber\\
&+\frac{1}{r^2\sin^2\theta_1\ldots\sin^2\theta_{N-3}}\frac{\partial^2\widetilde{F}}{\partial\phi^2}+\frac{\partial^2\widetilde{F}}{\partial z^2}.
\label{eq8.375b'}
\end{align}
In space ($N=3$), the Helmholtz equation in hypercylindrical coordinates simplifies to \eqref{eq8.349} in cylindrical coordinates, and in any dimension it can be written in a ``nested form'' that facilitates the subsequent solution by separation of variables.

Multiplying by $r^2$, the Helmholtz equation \eqref{eq8.375b'} is rewritten in a ``nested form'' as
\begin{align}
-k^2r^2\widetilde{F}&=r^2\nabla^2\widetilde{F}=r^2\left(\frac{\partial^2\widetilde{F}}{\partial r^2}+\frac{N-2}{r}\frac{\partial \widetilde{F}}{\partial r}\right)+\csc^{N-3}\theta_1\frac{\partial}{\partial\theta_1}\left(\sin^{N-3}\theta_1\frac{\partial \widetilde{F}}{\partial\theta_1}\right)\nonumber\\
&+\csc^2\theta_1\left[\csc^{N-4}\theta_2\frac{\partial}{\partial\theta_2}\left(\sin^{N-4}\theta_2\frac{\partial \widetilde{F}}{\partial\theta_2}\right)\right.\nonumber\\
&+\csc^2\theta_2\left[\csc^{N-5}\theta_3\frac{\partial}{\partial\theta_3}\left(\sin^{N-5}\theta_3\frac{\partial \widetilde{F}}{\partial\theta_3}\right)+\ldots\right.\nonumber\\
&+\csc^2\theta_{n-1}\left[\csc^{N-n-2}\theta_n\frac{\partial}{\partial\theta_n}\left(\sin^{N-n-2}\theta_n\frac{\partial \widetilde{F}}{\partial\theta_n}\right)+\ldots\right.\nonumber\\
&+\csc^2\theta_{N-4}\left[\csc\theta_{N-3}\frac{\partial}{\partial\theta_{N-3}}\left(\sin\theta_{N-3}\frac{\partial \widetilde{F}}{\partial\theta_{N-3}}\right)\right.\nonumber\\
&+\csc^2\theta_{N-3}\left.\frac{\partial^2 \widetilde{F}}{\partial\phi^2}\right]\ldots\left.\vphantom{\frac{\partial^2 \widetilde{F}}{\partial\phi^2}}\right]\ldots\left.\left.\vphantom{\frac{\partial^2 \widetilde{F}}{\partial\phi^2}}\right]\right]+r^2\frac{\partial^2\widetilde{F}}{\partial z^2}
\label{eq8.376'}
\end{align}
taking factors out of the brackets as early as possible. These equations in hypercylindrical coordinates can be obtained from the corresponding equations in hyperspherical coordinates making the transformation $N\rightarrow N-1$, noting also that there is one less latitude ($N-2\rightarrow N-3$), and an additional term must appear with regard to the coordinate $z$.

The solution of the Helmholtz equation in hypercylindrical coordinates is obtained by separation of variables,
\begin{equation}
\widetilde{F}\left(r, \theta_1, \ldots, \theta_{N-3}, \phi, z, \omega\right)=R\left(r\right)\, \Phi\left(\phi\right)\, Z\left(z\right)\prod_{n=1}^{N-3}\Theta_n\left(\theta_n\right),
\label{eq8.377'}
\end{equation}
with the substitution in \eqref{eq8.376'} and division by $\widetilde{F}$ leading to
\begin{align}
-k^2r^2&=\frac{r^2}{R}\left(\frac{\mathrm{d}^2R}{\mathrm{d}r^2}+\frac{N-2}{r}\frac{\mathrm{d}R}{\mathrm{d}r}\right)+\frac{1}{\Theta_1}\frac{\mathrm{d}^2\Theta_1}{\mathrm{d}\theta_1^2}+\left(N-3\right)\cot\theta_1\frac{1}{\Theta_1}\frac{\mathrm{d}\Theta_1}{\mathrm{d}\theta_1}\nonumber\\
&+\csc^2\theta_1\left[\frac{1}{\Theta_2}\frac{\mathrm{d}^2\Theta_2}{\mathrm{d}\theta_2^2}+\left(N-4\right)\cot\theta_2\frac{1}{\Theta_2}\frac{\mathrm{d}\Theta_2}{\mathrm{d}\theta_2}\right.\nonumber\\
&+\csc^2\theta_2\left[\frac{1}{\Theta_3}\frac{\mathrm{d}^2\Theta_3}{\mathrm{d}\theta_3^2}+\left(N-5\right)\cot\theta_3\frac{1}{\Theta_3}\frac{\mathrm{d}\Theta_3}{\mathrm{d}\theta_3}+\ldots\right.\nonumber\\
&+\csc^2\theta_{n-1}\left[\frac{1}{\Theta_n}\frac{\mathrm{d}^2\Theta_n}{\mathrm{d}\theta_n^2}+\left(N-n-2\right)\cot\theta_n\frac{1}{\Theta_n}\frac{\mathrm{d}\Theta_n}{\mathrm{d}\theta_n}+\ldots\right.\nonumber\\
&+\csc^2\theta_{N-4}\left[\frac{1}{\Theta_{N-3}}\frac{\mathrm{d}^2\Theta_{N-3}}{\mathrm{d}\theta_{N-3}^2}+\cot\theta_{N-3}\frac{1}{\Theta_{N-3}}\frac{\mathrm{d}\Theta_{N-3}}{\mathrm{d}\theta_{N-3}}\nonumber\right.\\
&+\csc^2\theta_{N-3}\left.\frac{1}{\Phi}\frac{\mathrm{d}^2\Phi}{\mathrm{d}\phi^2}\vphantom{\frac{\mathrm{d}^2\Theta_{N-3}}{\mathrm{d}\theta_{N-3}^2}}\right]\ldots\left.\vphantom{\frac{\mathrm{d}^2\Theta_{N-3}}{\mathrm{d}\theta_{N-3}^2}}\right]\ldots\left.\left.\vphantom{\frac{\mathrm{d}^2\Theta_{N-3}}{\mathrm{d}\theta_{N-3}^2}}\right]\right]+\frac{r^2}{Z}\frac{\mathrm{d}^2Z}{\mathrm{d}z^2},
\label{eq8.378'}
\end{align}
and separating the variables $\left(r, \theta_1, \theta_2, \ldots, \theta_n, \ldots, \theta_{N-3}, \phi, z\right)$ as much as possible, for the next step. Thus, the Helmholtz equation in hypercylindrical coordinates written in ``nested form'' \eqref{eq8.376'} has the solution by separation of variables \eqref{eq8.377'} leading to a set of $N$ separate ordinary differential equations, that are considered next.

The last term is the only one depending on the coordinate $z$ and it must be equal to a constant, namely $-K^2$, and results in the same equation as \eqref{eq8.351b}, whose solutions are the same sinusoidal functions as \eqref{eq8.352b}. The last but one term is the only one depending on the longitude $\phi$, so it must be a constant $-m^2$, leading to \eqref{eq8.351a}. It is exactly the same with regard to hyperspherical coordinates, comparing \eqref{eq8.351a} of $\phi$ with \eqref{eq17c}. The first and the last terms on the r.h.s. depend on the distance to the axis $r$, and so it equals to a constant, denoted by $q\left(1+q\right)$, leading to
\begin{equation}
r^2\frac{\mathrm{d}^2R}{\mathrm{d} r^2}+\left(N-2\right)r\frac{\mathrm{d} R}{\mathrm{d} r}+\left[\left(k^2-K^2\right)r^2-q\left(q+1\right)\right]R=0,
\label{eq8.379b'}
\end{equation}
that simplifies to \eqref{eq8.353b} for $N=3$ in cylindrical coordinates. The last latitude $\theta_{N-3}$ appears only in the last two terms in square brackets on the r.h.s. of \eqref{eq8.378'} and must be a constant leading to \eqref{eq17d} that on account of \eqref{eq8.351a} is equivalent to \eqref{eq17e}, performing the transformation $N\rightarrow N-1$. A similar reasoning for $\theta_{N-4}$ leads to \eqref{eq17f} with the same transformation. The corresponding ordinary differential equation for $\theta_{N-n}$, with $n=4,\ldots,N-1$, is equivalent to \eqref{eq17g}, but the factor $\left(n-1\right)$ must be replaced by $\left(n-2\right)$ in the second term of the l.h.s. (the first latitude corresponds to $n=N-1$). Regarding the latitudes, the parallelism from hyperspherical to hypercylindrical coordinates can be made considering the transformation $N\rightarrow N-1$. As in the case of hyperspherical coordinates, the constants introduced in the equations lead to $q_1=q$.

The solution by separation of variables of the Helmholtz equation in hypercylindrical coordinates \eqref{eq8.375b'} leads to a set of $N$ separate ordinary differential equations. The simplest dependence is on longitude and its solution is specified by \eqref{eq8.352a}, as in the cases of cylindrical and spherical coordinates, and also as in the case of hyperspherical coordinates observing the solution \eqref{eq18}, because the differential equation of $\phi$ is the same in all system of coordinates. Regarding the  radial dependence, the differential equation \eqref{eq8.379b'} is very similar to \eqref{eq17b} in the case of hyperspherical coordinates. The former equation can be obtained by the latter making the transformations on the constants $N\rightarrow N-1$ and $k^2\rightarrow k^2-K^2=\overline{k}^2$. In fact, the axial wavenumber $K$ is zero in hyperspherical coordinates. Therefore, the radial dependence in hypercylindrical coordinates can be obtained from the radial dependence in hyperspherical coordinates, that is, from \eqref{eq21b} and making the two transformations previously mentioned. Consequently, in hypercylindrical coordinates, it is also specified by a linear combination of Bessel functions, that is,
\begin{subequations}
\begin{equation}
R\left(r\right)=r^{\left(3-N\right)/2}Z_\sigma\left(\overline{k}r\right),
\end{equation}
but of order $\sigma$ equal to
\begin{equation}
\sigma^2=q\left(q+1\right)+\left(\frac{N}{2}-\frac{3}{2}\right)^2
\label{eq_sigma}
\end{equation}
\end{subequations}
and variable $\overline{k}r$ where $\overline{k}$ is again the radial wavenumber given by \eqref{eq8.354a}. The last latitude $\theta_{N-3}$ satisfies an associated Legendre differential equation, similar to \eqref{eq17e}, with order $m$ and degree $q_{N-3}$, thus leading to \eqref{eq23}, but making the transformation $N\rightarrow N-1$. For the remaining co-latitudes, a more general differential equation appears, similar to \eqref{eq24b}, that may be designated the hyperspherical associated Legendre differential equation. Its solution in terms of hyperspherical associated Legendre functions of two kinds specifies the dependence of the solution of the Helmholtz equation \eqref{eq8.343b} in hypercylindrical coordinates on the ($N-n$)-th latitude, with $n=4,\ldots,N-1$, given by \eqref{eq25d}. The only difference remains in the dimension because in hypercylindrical coordinates it is given by
\begin{equation}
2\lambda_n=n-3
\label{eq_lambda}
\end{equation}
and not as in \eqref{eq25c} for hyperspherical coordinates.

The simplest case of the solution of Helmholtz equation in hypercylindrical coordinates beyond cylindrical harmonics is four-dimensional, indicated in the appendix \ref{sec:appB.3}. 

\subsection{Hypercylindrical harmonics in four dimensions}
\label{sec:appB.3}

The four-dimensional case in hypercylindrical coordinates is similar to the three-dimensional case in hyperspherical coordinates. The relation to Cartesian coordinates, using \eqref{eq8.366f'-h'} and the first equation of \eqref{eq2}, involves an orthogonal Cartesian coordinate instead of a second latitude, that is,
\begin{align}
0\leq r<\infty, \quad 0\leq\theta\leq\pi, \quad 0\leq\phi\leq2\pi, \quad -\infty<z<+\infty: \nonumber\\
\left\lbrace w, x, y, z\right\rbrace=\left\lbrace r\cos\theta, r\sin\theta\cos\phi, r\sin\theta\sin\phi, z\right\rbrace. \label{eq8.407a-b}
\end{align}
The inverse coordinate transformation from four-dimensional Cartesian to hypercylindrical, knowing \eqref{eq8.367a'} to \eqref{eq8.367f'-g'}, is
\begin{align}
\begin{aligned}
r&=\left|w^2+x^2+y^2\right|^{1/2},\\
\cot\theta&=w\left|x^2+y^2\right|^{-1/2},\\
\cot\phi&=\frac{x}{y},\\
z&=z.
\end{aligned}
\label{eq8.408a-d}
\end{align}
The hypercylindrical base vectors, defined in \eqref{eq8.368a'-f'}, are
\begin{align}
\begin{aligned}
\overrightarrow{e}_r&=\left\lbrace\cos\theta, \sin\theta\cos\phi, \sin\theta\sin\phi, 0\right\rbrace,\\
\overrightarrow{e}_\theta&=r\left\lbrace-\sin\theta, \cos\theta\cos\phi, \cos\theta\sin\phi, 0\right\rbrace,\\
\overrightarrow{e}_\phi&=r\sin\theta\left\lbrace0,-\sin\phi,\cos\phi, 0\right\rbrace,\\
\overrightarrow{e}_z&=\left\lbrace0, 0, 0, 1\right\rbrace.
\end{aligned}
\label{eq8.409a-d}
\end{align}
They are mutually orthogonal and their modulus specify the scale factors \eqref{eq8.370a-b}, resulting in
\begin{subequations}
\begin{equation}
h_r=h_z=1, \quad h_\theta=r, \quad h_\phi=r\sin\theta,
\label{eq8.410a-d}
\end{equation}
as well as the determinant of the covariant metric tensor
\begin{equation}
\left|g\right|^{1/2}=r^2\sin\theta.
\label{eq8.410e}
\end{equation}
\end{subequations}
The four-dimensional Helmholtz equation in hypercylindrical coordinates, using \eqref{eq8.375b'}, is
\begin{align}
-k^2\widetilde{F}&=\nabla^2\widetilde{F}\nonumber\\
&=\frac{1}{r^2}\frac{\partial}{\partial r}\left(r^2\frac{\partial\widetilde{F}}{\partial r}\right)+\frac{1}{r^2\sin\theta}\frac{\partial}{\partial\theta}\left(\sin\theta\frac{\partial\widetilde{F}}{\partial\theta}\right)+\frac{1}{r^2\sin^2\theta}\frac{\partial^2\widetilde{F}}{\partial\phi^2}+\frac{\partial^2\widetilde{F}}{\partial z^2}
\label{eq8.411a,b}
\end{align}
where the Laplacian is similar to spherical \eqref{eq8.357}, changing the meaning of $r$ and adding the last term of the l.h.s. of the cylindrical equation \eqref{eq8.349}. Then, the solution by separation of variables
\begin{equation}
\widetilde{F}\left(r, \theta, \phi, z, \omega\right)=R\left(r\right)\,\Theta\left(\theta\right)\,\Phi\left(\phi\right)\,Z\left(z\right)
\label{eq8.412}
\end{equation}
is equal to the product of \eqref{eq8.363} by \eqref{eq8.352b} and using the last one taking into account the transformation $k^2\rightarrow k^2-K^2=\overline{k}^2$, that is,
\begin{align}
\widetilde{F}\left(r, \theta, \phi, z, \omega\right)&=\left(B_+\mathrm{e}^{\mathrm{i} Kz}+B_-\mathrm{e}^{-\mathrm{i} Kz}\right)\left(C_+\mathrm{e}^{\mathrm{i} m\phi}+C_-\mathrm{e}^{-\mathrm{i} m\phi}\right)\nonumber\\
&\times\left[D_+P_q^m\left(\cos\theta\right)+D_-Q_q^m\left(\cos\theta\right)\right]\left[E_+J_q\left(\overline{k}r\right)+E_-Y_q\left(\overline{k}r\right)\right].
\label{eq8.413}
\end{align}
It has been shown that the solution of the Helmholtz equation in four-dimensional hypercylindrical coordinates \eqref{eq8.411a,b} is \eqref{eq8.413} with arbitrary azimuthal $m$, radial $\overline{k}$ and axial $K$ wavenumbers. The superposition of solutions is valid and is used to specify the general integral of the equation of mathematical physics \eqref{eq13b} in hypercylindrical coordinates. The hyperspherical associated Legendre functions only appear for $N\geq4$ dimensions.

\subsection{Linear superposition of hyperspherical and hypercylindrical harmonics}
\label{sec:appB.4}

In the case of hyperspherical coordinates, when substituting the solution
\begin{equation}
\widetilde{F}\left(r, \theta_1, \ldots, \theta_{N-2}, \phi, \omega\right)=R\left(r\right)\,\Phi\left(\phi\right)\,\prod_{n=1}^{N-2}\Theta_n\left(\theta_n\right)
\label{eq8.377}
\end{equation}
of the Helmholtz equation \eqref{eq8.343b} in the generalized isotropic equation of mathematical physics \eqref{eq13b} using \eqref{eq8.342a}, the follows two changes are made: (i) a general superposition of solutions is used with integer parameters; (ii) only one function in each coordinate is used to simplify the formulas. Thus, the generalized isotropic equation of mathematical physics \eqref{eq13b} in hyperspherical coordinates has the solution as a superposition of hyperspherical harmonics (similar to the solution \eqref{eq26} in time domain)
\begin{align}
F\left(r, \theta_1, \ldots, \theta_{N-2}, \phi, t\right)=\int\limits_{-\infty}^{+\infty}\mathrm{d}\omega\,B_{m, q_1, \ldots, q_{N-2}}\left(\omega\right)\sum_{m=-\infty}^{+\infty}\mathrm{e}^{\mathrm{i}\left(m\phi-\omega t\right)}Z_\sigma\left(kr\right)\nonumber\\
\times\sum_{q_{N-2}=1}^\infty P_{q_{N-2}}^m\left(\cos\theta_{N-2}\right)\sum_{q_{N-3}, \ldots, q_1=1}^\infty\prod_{n=3}^{N-1}P_{q_{N-n}, \lambda_n}^{\mu_n}\left(\cos\theta_{N-n}\right)
\label{eq8.414}
\end{align}
involving: (i) an integration over frequency and arbitrary coefficients $B$ depending not only on the frequency, but also on the azimuthal $m$, radial $q=q_1$ and latitudinal $q_1, \ldots, q_{N-2}$ wavenumbers; (ii) the dependence on longitude $\phi$ that is specified by sinusoidal functions \eqref{eq18} with azimuthal wavenumber $m$; (iii) the dependence on the radius $r$ that is specified by Bessel functions \eqref{eq21b} of order \eqref{eq20a} and variable $kr$ where $k$ is the wavenumber defined by \eqref{eq8.343a}; (iv) the dependence on the last latitude $\theta_{N-2}$ specified by associated Legendre functions \eqref{eq23} of order $m$ and degree $q_{N-2}$; (v) the dependence on the remaining $N-3$ latitudes $\theta_1, \ldots, \theta_{N-3}$ specified by hyperspherical associated Legendre functions \eqref{eq25d} with degree \eqref{eq25a}, order \eqref{eq25b} and dimension \eqref{eq25c} knowing that $n=3, \ldots, N-1$. The sums over the wavenumbers $m, q_1, \ldots, q_{N-2}$ could be replaced by integrals. Similar results apply in hypercylindrical coordinates.

In the case of hypercylindrical coordinates, the Laplacian
\eqref{eq8.375b'}, compared to hyperspherical coordinates, has one less latitude replaced by an axial Cartesian coordinate in the generalized equation of mathematical physics \eqref{eq13b} leading to
\begin{align}
0&=k^2 \widetilde{F}+\frac{1}{r^{N-2}}\frac{\partial}{\partial r}\left(r^{N-2}\frac{\partial\widetilde{F}}{\partial r}\right)+\frac{1}{r^2\sin^{N-3}\theta_1}\frac{\partial}{\partial\theta_1}\left(\sin^{N-3}\theta_1\frac{\partial\widetilde{F}}{\partial\theta_1}\right)\nonumber\\
&+\frac{1}{r^2\sin^{2}\theta_1\sin^{N-4}\theta_2}\frac{\partial}{\partial\theta_2}\left(\sin^{N-4}\theta_2\frac{\partial\widetilde{F}}{\partial\theta_2}\right)+\ldots\nonumber\\
&+\frac{1}{r^2\sin^{2}\theta_1\ldots\sin^{2}\theta_{n-1}\sin^{N-n-2}\theta_n}\frac{\partial}{\partial\theta_n}\left(\sin^{N-n-2}\theta_n\frac{\partial\widetilde{F}}{\partial\theta_n}\right)+\ldots\nonumber\\
&+\frac{1}{r^2\sin^{2}\theta_1\ldots\sin^2\theta_{N-4}\sin\theta_{N-3}}\frac{\partial}{\partial\theta_{N-3}}\left(\sin\theta_{N-3}\frac{\partial\widetilde{F}}{\partial\theta_{N-3}}\right)\nonumber\\
&+\frac{1}{r^2\sin^2\theta_1\ldots\sin^2\theta_{N-3}}\frac{\partial^2\widetilde{F}}{\partial\phi^2}+\frac{\partial^2\widetilde{F}}{\partial z^2}.
\label{eq8.415}
\end{align}
The solution is similar to \eqref{eq8.414} with one less latitude, specifically $\theta_{N-2}$, replaced by an axial Cartesian coordinate $z$, whose its solution is given by \eqref{eq8.352b}, and that leads to
\begin{align}
F\left(r, \theta_1, \ldots, \theta_{N-3}, \phi, z, t\right)&=\int\limits_{-\infty}^{+\infty}\int\limits_{-\infty}^{+\infty}\mathrm{d} K\,\mathrm{d}\omega\,B_{m, q_{1}, \ldots, q_{N-3}}\left(K, \omega\right)\nonumber\\
&\times\sum_{m=-\infty}^{+\infty}\mathrm{e}^{\mathrm{i}\left(Kz+m\phi-\omega t\right)}Z_\sigma\left(\overline{k}r\right)\nonumber\\
&\times\sum_{q_{N-3}=1}^{\infty}P_{q_{N-3}}^m\left(\cos\theta_{N-3}\right)\nonumber\\
&\times\sum_{q_{N-4}, \ldots, q_{1}=1}^\infty\prod_{n=4}^{N-1}P_{q_{N-n}, \lambda_n}^{\mu_n}\left(\cos\theta_{N-n}\right). \label{eq8.416}
\end{align}
Thus, the solution of the generalized isotropic equation of mathematical physics in hypercylindrical coordinates \eqref{eq8.416} is similar compared to the solution in hyperspherical coordinates \eqref{eq8.414}, with the following distinctions: (i) the last latitude is specified by associated Legendre functions, therefore $\theta_{N-2}$ for hyperspherical coordinates, however $\theta_{N-3}$ for  hypercylindrical coordinates; (ii) in both systems, all others latitudes are specified by hyperspherical associated Legendre functions, however their dimensions are different, that is, they are given by \eqref{eq25d} for hyperspherical or \eqref{eq_lambda} for hypercylindrical coordinates; (iii) in both systems, there is the integration of frequency $\omega$, however in hypercylindrical coordinates there is also of the axial wavenumber $K$, and the coefficient $B$ is a function of two variables (and not one) with azimuthal $m$, radial $q=q_1$ and $N-4$ (and not $N-3$) latitudinal $q_{2}, \ldots, q_{N-3}$ wavenumbers; (iv) in this last case, the radial dependence is specified by Bessel functions of order \eqref{eq_sigma} and variable $\overline{k}r$ involving the distance from the axis $r$ and the radial wavenumber $\overline{k}\neq k$. These results involve special functions, for example, Bessel and hyperspherical associated Legendre functions that are solutions of linear differential equations with variable coefficients. 

\providecommand{\href}[2]{#2}

\bibliographystyle{unsrt}
\bibliography{ms}

\address{
CCTAE, IDMEC, LAETA, Instituto Superior T\'{e}cnico, Universidade de Lisboa,\\
Av. Rovisco Pais 1, 1049-001 Lisboa, Portugal\\
\email{luis.campos@tecnico.ulisboa.pt}\\
\email{manuel.jose.dos.santos.silva@tecnico.ulisboa.pt}
\\
}

\end{document}